\documentclass[prl,aps,superscriptaddress]{revtex4}
\usepackage{bm}
\usepackage{amsfonts}
\usepackage{amsthm}
\usepackage[dvips]{graphicx}
\usepackage{mathrsfs}
\usepackage[intlimits]{amsmath}
\usepackage[colorlinks, citecolor=red]{hyperref}

\newtheorem{theorem}{Theorem}

\newtheorem{definition}{Definition}

\newtheorem{lemma}{Lemma}

\newcommand{\bra}[1]{\langle#1|}
\newcommand{\ket}[1]{|#1\rangle}
\newcommand{\innerp}[2]{\langle#1|#2\rangle}

\newcommand{\poly}{\mbox{poly}}

\begin{document}

\title{An efficient quantum algorithm for generative machine learning}
\author{X. Gao$^{1}$, Z.-Y. Zhang$^{1,2}$, and L.-M. Duan}
\affiliation{Center for Quantum Information, IIIS, Tsinghua University, Beijing 100084,
PR China}
\affiliation{Department of Physics, University of Michigan, Ann Arbor, Michigan 48109, USA}

\begin{abstract}
A central task in the field of quantum computing is to find
applications where quantum computer could provide exponential speedup over
any classical computer \cite{shor1999polynomial,feynman1982simulating,lloyd1996universal}. Machine learning represents an important field
with broad applications where quantum computer may offer significant speedup
\cite{Biamonte2017,ciliberto2017quantum,brandao2016quantum,brandao2017exponential,farhi2001quantum}.
Several quantum algorithms for discriminative machine
learning \cite{jebara2012machine} have been found based on efficient solving of linear algebraic problems \cite{harrow2009quantum,wiebe2012quantum,lloyd2013quantum,lloyd2014quantum,rebentrost2014quantum,cong2016quantum},
with potential exponential speedup in runtime under the assumption of
effective input from a quantum random access memory \cite{giovannetti2008quantum}. In machine learning,
generative models represent another large class \cite{jebara2012machine} which is widely used for
both supervised and unsupervised learning \cite{shalev2014understanding,goodfellow2016deep}.
Here, we propose an efficient quantum algorithm for machine learning based on a quantum generative model.
We prove that our proposed model is exponentially more powerful to represent
probability distributions compared with classical generative models and has
exponential speedup in training and inference at least for some
instances under a reasonable assumption in computational complexity
theory. Our result opens a new direction for quantum machine learning and
offers a remarkable example in which a quantum algorithm shows exponential
improvement over any classical algorithm in an important application
field. 
\end{abstract}

\maketitle

Machine learning and artificial intelligence represent a very important
application area which could be revolutionized by quantum computers with clever algorithms that offer
exponential speedup \cite
{Biamonte2017,ciliberto2017quantum}. The candidate algorithms with potential
exponential speedup so far rely on efficient quantum solution of linear
system of equations or linear algebraic problems \cite%
{lloyd2013quantum,lloyd2014quantum,rebentrost2014quantum,cong2016quantum}.
Those algorithms require quantum random access memory (QRAM) as a
critical component in addition to a quantum computer. In a QRAM,
the number of required quantum routers scales up exponentially with the
number of qubits in those algorithms \cite{giovannetti2008quantum,giovannetti2008architectures}. This exponential overhead in
resource requirement poses a significant challenge for its experimental
implementation and is a caveat for fair comparison with corresponding classical
algorithms \cite{aaronson2015read,ciliberto2017quantum}.

\begin{figure}[tbp]
\includegraphics[width=1\linewidth]{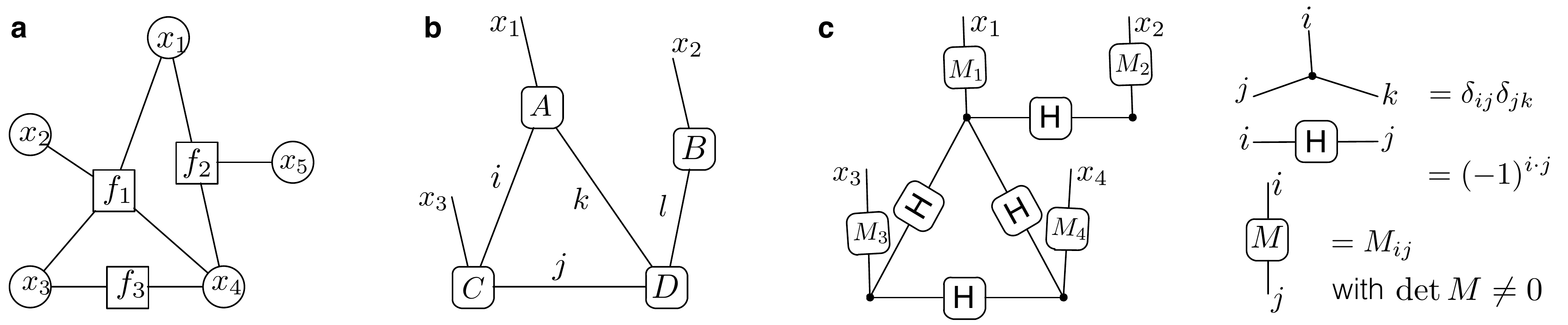}
\caption{\textbf{Classical and quantum generative models.} \textbf{a},
Illustration of a factor graph, which includes widely-used classical generative models as its
special cases. A factor graph is a bipartite graph where one group of the vertices represent variables (denoted by circles)
and the other group of vertices represent positive functions (denoted by squares) acting on connected variables.
The corresponding probability distribution is given by the product of all these functions. For instance, the probability distribution in (a) is $%
p(x_{1},x_{2},x_{3},x_{4},x_{5})=f_{1}(x_{1},x_{2},x_{3},x_{4})f_{2}(x_{1},x_{4},x_{5})f_{3}(x_{3},x_{4})/Z
$ where $Z$ is a normalization factor. Each variable connects to at most a constant number of functions which
introduce correlations in the probability distribution. \textbf{b,} Illustration of a tensor network state.
Each unshared (shared) edge represents a physical (hidden) variable, and each vertex represents a complex function of the variables on its connected edges. The wave function of the physical variables is defined as a product of the functions on all the vertices, after summation (contraction) of the hidden variables. Note that a tensor network state can be regarded as a quantum version of the factor graph after partial contraction (similar to marginal probability in classical case) with positive real functions replaced by complex functions. \textbf{c,} Definition of a quantum generative model (QGM) introduced in this paper. The state is a special type of tensor network state, with the vertex functions fixed to be three types as shown on the right side. Without the single-qubit invertible matrix $M_i$ which contains the model parameters, the wave function connected by Hadamard and identity matrices just represent a graph state. To get a probability distribution from this model, we measure a subset of $n$ qubits (among total $m$ qubits corresponding to physical variables) in the computational basis under this state. The unmeasured $m-n$ qubits are traced over to get the marginal probability distribution $P(\{x_i\})$ of the measured $n$ qubits. We prove in this paper that the $P(\{x_i\})$ is general enough to include probability distributions of all classical factor graphs and special enough to allow a convenient quantum algorithm for the parameter training and inference.}
\label{fig:generative_models}
\end{figure}

In this paper, we propose a quantum algorithm with potential exponential
speedup for machine learning based on generative models. Generative models
are widely used to learn the underlying probability distribution
describing correlations in observed data. It has utility in both
supervised and unsupervised learning with a wide range of applications from
classification, feature extraction, to creating new data such as style transfer
\cite{shalev2014understanding,goodfellow2016deep,bishop2006pattern}.
Compared with discriminative models such as support vector machine or
feed-forward neural network, generative models can express much more complex
relations among variables \cite{jebara2012machine}, which makes them broadly applicable but
at the same time harder to tackle. Typical generative models include
probabilistic graphical models such as Bayesian network and Markov
random field \cite{bishop2006pattern}, and generative neural networks such
as Boltzmann machine and deep belief network. All these classical
probabilistic models can be transformed into the so-called
factor graphs \cite{bishop2006pattern}.

Here we introduce a quantum generative model (QGM) where the probability
distribution describing correlations in data is generated by measuring a set of observables
under a many-body entangled state.  A generative model is largely characterized
by its representational power and its performance to learn the model parameters
from the data and to make inference about complex relationship between any variables. In terms of representational power, we
prove that our introduced QGM can efficiently represent any factor graphs, which
include almost all the classical generative models in practical applications as particular cases.
Furthermore, we show that the QGM is exponentially more
powerful than factor graphs by proving that at least
some instances generated by the QGM cannot be efficiently represented by any
factor graph with polynomial number of variables if a widely accepted
conjecture in computational complexity theory holds, that is, the polynomial
hierarchy, which is a generalization of the famous P versus NP problem, does not collapse.

Representational power and generalization ability \cite{sm} only measure
one aspect of a generative model. On the other hand we need an effective algorithm
for training and making inference. We propose a general learning algorithm utilizing quantum phase estimation of the
constructed parent Hamiltonian for the underlying many-body entangled state.
Although it is unreasonable to expect that the proposed quantum algorithm
has polynomial scaling in runtime in all cases (as this implies ability of a
quantum computer to efficiently solve any NP problem, an unlikely result),
we prove that at least for some instances, our quantum algorithm has
exponential speedup over any classical algorithm, assuming quantum
computers cannot be efficiently simulated by classical computers, a conjecture which is
believed to hold.

The intuition for quantum speedup in our algorithm can be
understood as follows: the purpose of generative machine learning is to
model any data generation process in nature by finding the underlying
probability distribution. As nature is governed by the law of quantum
mechanics, it is too restrictive to assume that the real world data can
always be modelled by an underlying probability distribution as in
classical generative models. Instead, in our quantum generative model, correlation in data is parameterized by
the underlying probability amplitudes of a many-body entangled
state. As the interference of quantum probability amplitudes can lead to
phenomena much more complex than those from classical probabilistic models,
it is possible to achieve big improvement in our quantum generative model under
certain circumstances.

We start by defining factor graph and our QGM. Direct characterization
of a probability distribution of $n$ binary variables has an exponential
cost of $2^{n}$. A factor graph, which includes many classical generative
models as special cases, is a compact way to represent $n$-particle
correlation \cite{bishop2006pattern,goodfellow2016deep}.
As shown in Fig. \ref{fig:generative_models}a, a factor graph is associated with a bipartite graph where the
probability distribution can be expressed as a product of positive
correlation functions of a constant number of variables. Here, without loss
of generality, we assumed constant-degree graph, in which the maximum number of edges per vertex is bounded by a constant.

Our QGM is defined on a graph state $|G\rangle $ of $m$ qubits associated
with a graph $G$. We introduce the following transformation

\begin{equation}
|Q\rangle \equiv M_{1}\otimes \cdots \otimes M_{m}|G\rangle ,
\end{equation}%
where $M_{i}$ denotes an invertible (in general non unitary) $2\times 2$
matrix applied on the Hilbert space of qubit $i$. From $m$ vertices of the graph $G$, we choose a
subset of $n$ qubits as the visible units and measure them in computational basis $\left\{ \left\vert0\right\rangle ,\left\vert 1\right\rangle \right\} $. The measurement outcomes sample from a probability
distribution $Q\left( \left\{ x_{i}\right\} \right) $ of $n$ binary variables $\left\{ x_{i},i=1,2,\cdots n\right\} $ (the other $m-n$ hidden qubits are just traced over to give the
reduced density matrix). Given graph $G$ and the subset of visible
vertices, the distribution $Q\left( \left\{ x_{i}\right\} \right) $ defines
our QGM which is parameterized efficiently by the parameters in the
matrices $M_{i}$. The state $|Q\rangle $ can be written as a special tensor
network state (see Fig. \ref{fig:generative_models}) \cite{sm}. We define our model in this form for two reasons:
first, the probability distribution $Q\left( \left\{ x_{i}\right\} \right) $ needs to be general enough to include all factor graphs; second, if the state $|Q\rangle$ takes a specific form, the parameters in this model can be conveniently trained by a quantum algorithm on a data set.

\begin{figure}[tbp]
\includegraphics[width=1\linewidth]{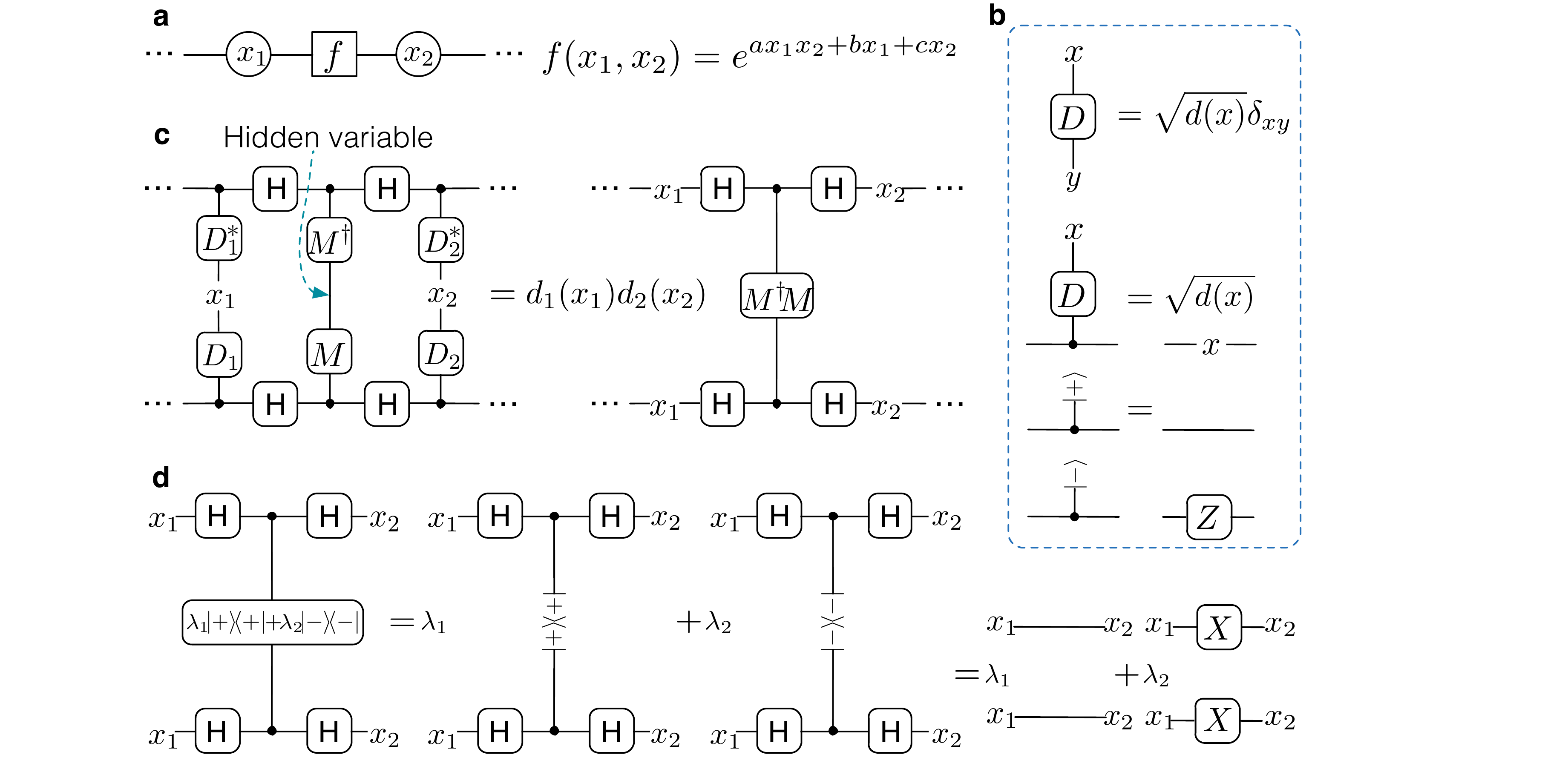}
\caption{ \textbf{Efficient representation of factor graphs by the quantum generative model.} \textbf{a,}
The general form of correlation function of two binary variables in a factor graph, with parameters $a,b,c$ being real.
This correlation acts as the building block for general
correlations in any factor graph by use of the universal approximation theorem \cite{le2008representational}. \textbf{b,} Notations of some common tensors and their identities: $D$ is a diagonal matrix with diagonal elements $\protect%
\sqrt{d(x})$ with $x=0,1$; $Z$ is the diagonal Pauli matrix $\text{diag}(1,-1)$; and $|\pm \rangle =(|0\rangle \pm |1\rangle )/%
\protect\sqrt{2}$. \textbf{c,} Representation of the building block correlator $f(x_1, x_2)$ in a factor
graph (see \textbf{a}) by the QGM with one hidden variable (unmeasured) between two visible variables $x_1, x_2$ (measured in the computational basis).  We choose the single-bit matrix $D_1, D_2$ to be diagonal with $D_{1}=\text{diag}(%
\protect\sqrt{d_{1}(0)},\protect\sqrt{d_{1}(1)})$ and $D_{2}=\text{diag%
}(\protect\sqrt{d_{2}(0)},\protect\sqrt{d_{2}(1)})$. In simplification of this graph, we used the identity in \textbf{b}. \textbf{d,} Further simplification of the graph in \textbf{c}, where we choose the form of the single-bit matrix $M^{\dag }M$ acting on the hidden variable to be $M^{\dag }M=\protect\lambda _{1}|+\rangle \langle +|+\protect%
\lambda _{2}|-\rangle \langle -|$ with positive eigenvalues $\lambda _{1},\lambda _{2}$. We used the identity in \textbf{b} and the relation $HZH=X$, where X (H) denotes the Pauli (Hadamard) matrix, respectively. By solving the values of $\lambda _{1},\lambda _{2},d_1(x_1),d_2(x_2)$ in terms of $a,b,c$ (see the proof of Theorem \ref{theorem1}), this correlator of QGM exactly reproduces
the building block correlator $f(x_1, x_2)$ of the factor graph.   }
\label{fig:Q_model}
\end{figure}

Now we show that any factor graph can be viewed as a special case of QGM by the
following theorem:

\begin{theorem}
\emph{The QGM defined above can efficiently represent probability distributions from any constant-degree factor graphs
by an arbitrarily high precision. }
\label{theorem1}
\end{theorem}

As probabilistic graphical models and generative neural networks can
all be reduced to constant-degree factor graphs \cite{bishop2006pattern,sm}, the above theorem shows that our proposed QGM is general enough to include those probability distributions in widely-used classical generative models.

\textbf{Proof of Theorem \ref{theorem1}}: First, for any factor graph with degree bounded by a constant $k$, by the universal approximation theorem \cite%
{le2008representational}, each $k$-variable node function can be
approximated arbitrarily well with $2^{k}+k$ variables ($2^{k}$ of them are
hidden) connected by the two-variable correlator that takes the
generic form $f(x_{1},x_{2})=e^{ax_{1}x_{2}+bx_{1}+cx_{2}}$, where $x_{1}$, $x_{2}$ denote the binary variables and $a,b,c$ are real parameters. As $Q\left( \left\{ x_{i}\right\} \right) $ has a similar factorization structure as the factor graph after measuring the visible qubits $x_{i}$ under a diagonal matrix $M_{i}$ (see Fig. \ref{fig:Q_model}), it is sufficient to show that each correlator $f(x_{1},x_{2})$ can be constructed in the QGM. This construction can be achieved by adding one hidden variable (qubit) $j$ with invertible matrix $M_{j}$ between two visible variables $x_{1}$ and $x_{2}$. As shown in Fig. \ref{fig:Q_model}, we take $M_{1}$ and $M_{2}$ to be diagonal with eigenvalues $\sqrt{d_{1}\left( x_{1}\right) }$ and $\sqrt{d_{2}\left(x_{2}\right) }$, respectively, and $M_{j}^{\dagger}M_{j}=\lambda_{1}\left\vert +\right\rangle \left\langle+\right\vert +\lambda_{2}\left\vert -\right\rangle \left\langle -\right\vert $, where $\left\vert\pm \right\rangle =\left( \left\vert 0\right\rangle \pm \left\vert
1\right\rangle \right) /\sqrt{2}$. The correlator between $x_{1}$ and $x_{2}$
in the QGM is then given by $d_{1}\left( x_{1}\right) d_{2}\left(x_{2}\right) \left[ \lambda _{1}\delta _{x_{1}x_{2}}+\lambda _{2}(1-\delta_{x_{1}x_{2}})\right] /2$. We want it to be equal to $e^{ax_{1}x_{2}+bx_{1}+cx_{2}}$ to simulate the factor graph. There exists a simple solution with $d_{1}\left( 0\right)=d_{2}\left(0\right) =\lambda _{1}/2=1$ and $d_{1}\left( 1\right) =e^{b+a/2},$ $d_{2}\left( 1\right) =e^{c+a/2},$ $\lambda _{2}=2e^{-a/2}$. This completes the proof.

Furthermore, we can show that the QGM is exponentially more powerful than factor graphs in representing probability distributions. This is summarized by the following theorem:

\begin{theorem}
\emph{If the polynomial hierarchy in the computational complexity theory
does not collapse, there exist probability distributions that can be
efficiently represented by a QGM but cannot be efficiently represented even
under approximation by conditional probabilities from any classical
generative models that are reducible to factor graphs. }
\label{theorem2}
\end{theorem}The proof of this theorem involves many terminologies and results from the
computational complexity theory, so we present it in the Supplementary Material \cite{sm}.

\begin{figure}[b]
\includegraphics[width=1\linewidth]{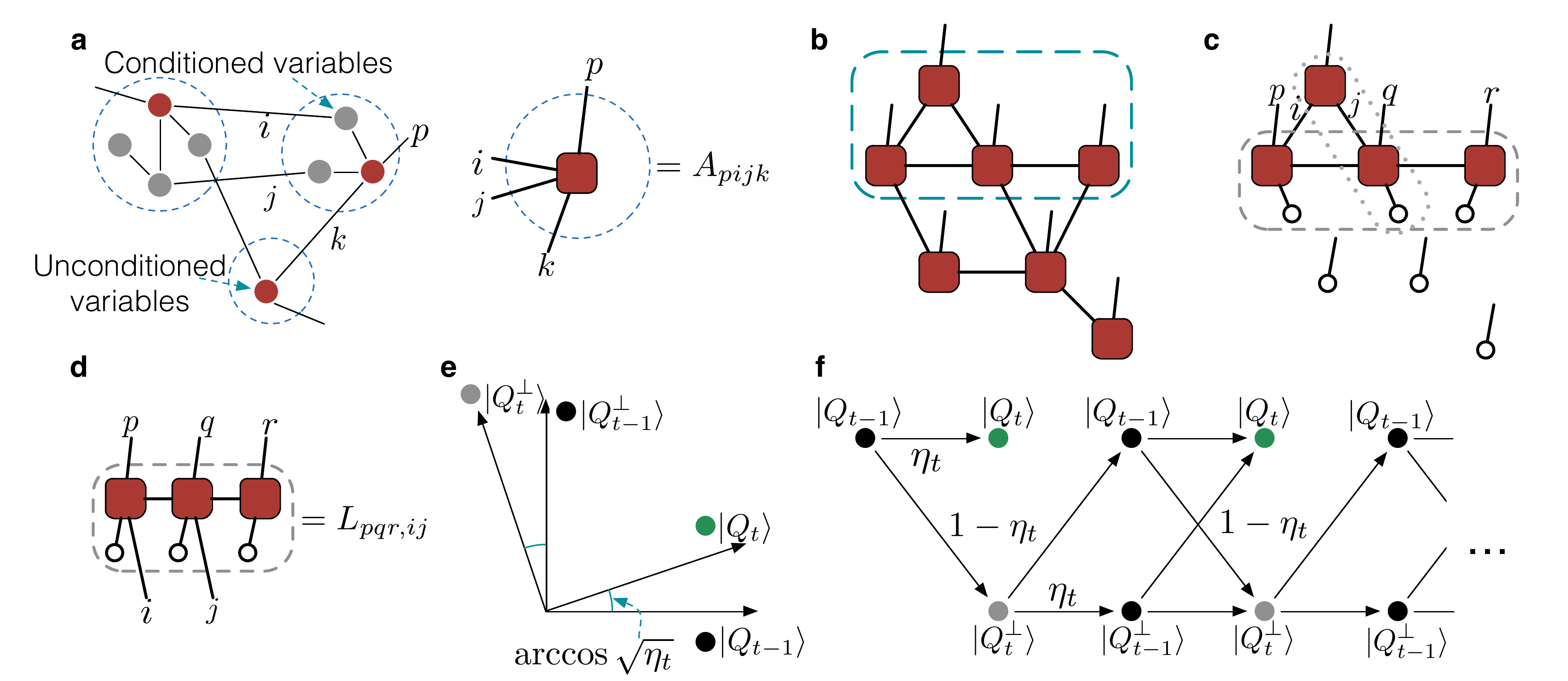}
\caption{ \textbf{Illustration of our training algorithm for the quantum generative model.}
\textbf{a,} Training and inference of the QGM are reduced to measuring certain operators under the state $|Q(z)\rangle$. The key step of the quantum algorithm
is therefore to prepare the state $|Q(z)\rangle$, which is achieved by recursive quantum phase estimation of the constructed parent Hamiltonian. The variables in the set $z$ whose values are specified are called conditioned variables, whereas the other variables that carry the binary physical index are called unconditioned variables. We group the variables in a way such that each group contains only one unconditioned variable and different groups are connected by a small constant number of edges (representing virtual indices or hidden variables). Each group then defines a tensor with one physical index (denoted by $p$) and a small constant number of virtual indices (denoted by $i,j,k$ in the figure). \textbf{b, } Tensor network representation of $|Q(z)\rangle$, where a local tensor is defined
for each group specified in \textbf{a}. \textbf{c, } Tensor network representation of $|Q_{t}\rangle $, where $|Q_{t}\rangle $ are the series of states
reduced from $|Q_{z}\rangle $. In each step of the reduction, one local tensor is moved out. The moved-out local tensors are represented by unfilled circles, each carrying a physical index set to $0$. For the edges between the remaining tensor network and the moved-out tensors, we set the corresponding virtual indices to $0$ (represented by unfilled circles). \textbf{d, } Construction of the parent Hamiltonian. The figure shows how to construct one term in the parent Hamiltonian, which corresponds to a group of neighboring local tensors such as those in the dashed box in \textbf{c}. After contraction of all virtual indices among the group, we get a tensor $L_{pqr,ij}$, which defines a linear map $L$ from virtual indices $i,j$ to physical indices $p,q,r$. As the indices $i,j$ take all the possible values, the range of this mapping $L$ spans a subspace $\text{range}(L)$ in the Hilbert space $H_{p,q,r}$ of the physical indices $p,q,r$. This subspace has a complementary orthogonal subspace inside $H_{p,q,r}$, denoted by $\text{comp}(L)$. The projector to the subspace $\text{comp}(L)$ then defines one term in the parent Hamiltonian, and by this definition $|Q_{t}\rangle $ lies in the kernel space of this projector. We construct each local term with a group of neighboring tensors. Each local tensor can be involved in several Hamiltonian terms (as illustrated in \textbf{c} by the dashed box and the dotted box), thus some neighboring groups have non-empty overlap, and they generate terms that in general do not commute. By this method, one can construct the parent Hamiltonian whose ground state uniquely defines the state $|Q_{t}\rangle $ \cite{perez2008peps}. \textbf{e, } States involved in the evolution from  $|Q_{t-1}\rangle $  to $|Q_{t}\rangle $
by quantum phase estimation applied on their parent Hamiltonians. $|Q_{t-1}^{\perp }\rangle ,|Q_{t}^{\perp }\rangle $ represent
the states orthogonal to $|Q_{t-1}\rangle ,|Q_{t}\rangle $, respectively, inside the two-dimensional subspace spanned by  $|Q_{t-1}\rangle $  and $|Q_{t}\rangle $. The angle between $|Q_{t}\rangle $ and $|Q_{t-1}\rangle $ is determined by the overlap $\protect\eta _{t}=|\langle Q_{t}|Q_{t-1}\rangle |^{2}$.
\textbf{f, } State evolution under recursive application of quantum phase estimation algorithm. Starting from the state $|Q_{t-1}\rangle $, we always stop at the state $|Q_{t}\rangle $, following any branch of this evolution, where $\eta _{t}$ and $1-\eta _{t}$ denote the probabilities of the corresponding outcomes. }
\label{fig:parent_hamiltonian}
\end{figure}

For a generative model to be useful for machine learning, apart from high
representational power and generalization ability (which are closely related  \cite{sm}), we also need to have efficient algorithms for both inference and training. For inference, we usually need to compute conditional
probability $\sum_{y}p(x,y|z)$ \cite{bishop2006pattern}, where $x,y,z$
denote variable sets. For training, we choose to minimize the
Kullback--Leibler (KL) divergence $D(q_{\text{d}}||p_{\theta })=-\sum_{v}q_{%
\text{d}}(v)\log (p_{\theta }(v)/q_{\text{d}}(v))$ between $q_{\text{d}}$,
the distribution of the given data sample, and $p_{\theta }$, the distribution of the generative model, with the whole parameter set denoted by $\theta $. Typically, one
minimizes $D(q_{\text{d}}||p_{\theta })$ by optimizing the model parameters $\theta$ using the gradient descent method \cite{goodfellow2016deep}. The $\theta $-dependent part $D\left( \theta \right) $ of $D(q_{\text{d}}||p_{\theta })$ can be expressed as $-\sum_{v\in \text{data set}}\log p_{\theta }(v)/M$, where $M$ denotes the total number of data points. As the number of parameters is bounded by $\mbox{poly}(n)$,
the required data size $M$ for training is typically also bounded by a $\mbox{poly}(n)$ function \cite{sm,shalev2014understanding}.

In our QGM, both of the conditional probability $\sum_{y}p(x,y|z)$ and the gradient of KL divergence $\partial _{\theta }D\left( \theta \right) $ can be conveniently calculated using the structure of state $|Q\rangle $ defined in Fig. \ref{fig:generative_models}. We first define a tensor network state $|Q(z)\rangle
\equiv (I\otimes \langle z|)|Q\rangle $ by projecting the variable set $z$
to the computational basis. As shown in the Supplementary Material \cite{sm}, the
conditional probability can be expressed as

\begin{equation}
\sum_{y}p(x,y|z)=\frac{\langle Q(z)|O|Q(z)\rangle }{\langle Q(z)|Q(z)\rangle
},
\label{QGM_rep}
\end{equation}%
which is the expectation value of the
operator $O=|x\rangle \langle x|$ under the state $|Q(z)\rangle $. Similarly, we show in the Supplementary Material \cite{sm} that $\partial _{\theta }D\left( \theta
\right) $ can be reduced to a combination of terms taking the same form as Eq.
\ref{QGM_rep}, with operator $O$ replaced by $O_{1}=(\partial _{\theta
_{i}}M_{i})M_{i}^{-1}+H.c.$ or $O_{2}=|v_{i}\rangle \langle v_{i}|(\partial
_{\theta _{i}}M_{i})M_{i}^{-1}+H.c.$, where $\theta _{i}$ denotes a specific
parameter in the invertible matrix $M_{i}$; $v_{i}$ is the qubit of data $v$
corresponding to variable $x_{i}$; and $H.c.$ stands for the Hermitian conjugate
term. The variable $z$ in this case takes the value of $v$ (or $v$ excluding
$v_{i}$) expressed in a binary string.

With the above simplification, training or inference in the QGM is reduced to preparation of the tensor network state $|Q(z)\rangle $. With an algorithm similar to the one in Ref. \cite{schwarz2012preparing}, we use recurrent quantum phase estimation to prepare the state $|Q(z)\rangle $. For this purpose, first we construct the parent Hamiltonian $H(z)$ whose unique ground state is $|Q(z)\rangle $. This is shown in Fig. \ref{fig:parent_hamiltonian}, where the variables $z=\{z_{i}\}$ are grouped as in Fig. \ref{fig:parent_hamiltonian}a. In this case, the corresponding local tensors in the tensor network state $|Q(z)\rangle $ are all easy to compute and of constant degree. The parent Hamiltonian $H(z)$ is constructed through contraction of these local tensors as shown in Fig. \ref{fig:parent_hamiltonian}c to 3d \cite{perez2008peps}.

By construction of the parent Hamiltonian for the tensor network state,
the quantum algorithm for training and inference in the QGM is realized through the following steps:

\begin{description}
\item \textbf{Step 1:} Construct a series of tensor network states $\{|Q_{t}\rangle \}$ with $t=0,1,... ,n$ as shown in Fig. \ref{fig:parent_hamiltonian}c by reduction from $|Q_{n}\rangle =|Q(z)\rangle $. The initial state $|Q_{0}\rangle $ is a product state $|0\rangle ^{\otimes n}$, and $|Q_{t}\rangle $ is constructed from $|Q_{t-1}\rangle $ by adding one more tensor in $|Q(z)\rangle $ that is not contained in $|Q_{t-1}\rangle $ and setting the uncontracted virtual indices as $0$.

\item \textbf{Step 2:} Construct a parent Hamiltonian $H_{t}$ for each $|Q_{t}\rangle $ with the method illustrated in Fig. \ref{fig:parent_hamiltonian} \cite{perez2008peps}.

\item \textbf{Step 3:} Starting from $|Q_{0}\rangle $, we prepare $|Q_{1}\rangle , ... ,|Q_{n}\rangle $ sequentially. Suppose we have
prepared $|Q_{t-1}\rangle $, the following sub-steps will prepare $|Q_{t}\rangle $ based on the recursive quantum phase estimation. 

\begin{description}
\item \textbf{Sub-step 1:} Use the phase estimation algorithm \cite{abrams1999quantum} on the parent Hamiltonians $H_{t}$ and $H_{t-1}$ to
implement two projective measurements of $\{|Q_{t}\rangle \langle
Q_{t}|,I-|Q_{t}\rangle \langle Q_{t}|\}$ and $\{|Q_{t-1}\rangle \langle
Q_{t-1}|,I-|Q_{t-1}\rangle \langle Q_{t-1}|\}$ (see Fig. \ref{fig:parent_hamiltonian}e). The runtime of this algorithm is proportional
to $1/\Delta _{t}$ or $1/\Delta _{t-1}$ with some $\mbox{poly}(n)$ overhead, where $\Delta _{t}$ ($\Delta _{t-1}$) denotes the energy gap of the Hamiltonian $H_{t}$ ($H_{t-1}$).

\item \textbf{Sub-step 2:} Starting from $|Q_{t-1}\rangle $, we perform the
projective measurement $\{|Q_{t}\rangle \langle Q_{t}|,I-|Q_{t}\rangle
\langle Q_{t}|\}$. If the result is $|Q_{t}\rangle $, we succeed and skip the following sub-steps. Otherwise, we get $|Q_{t}^{\perp }\rangle $ lying in the plane spanned by $|Q_{t-1}\rangle $ and $|Q_{t}\rangle $.

\item \textbf{Sub-step 3:} We perform the projective measurement $\{|Q_{t-1}\rangle \langle Q_{t-1}|,I-|Q_{t-1}\rangle \langle Q_{t-1}|\}$ on the state $|Q_{t}^{\perp }\rangle $. The result is either $|Q_{t-1}\rangle $ or $|Q_{t-1}^{\perp }\rangle $.

\item \textbf{Sub-step 4:} We perform the projective measurement $\{|Q_{t}\rangle \langle Q_{t}|,I-|Q_{t}\rangle \langle Q_{t}|\}$ again. We either succeed in getting $|Q_{t}\rangle $, with probability $\eta _{t}=|\langle Q_{t}|Q_{t-1}\rangle |^{2}$, or have $|Q_{t}^{\perp }\rangle $. In the latter case, we go back to the sub-step 3 and continue until success.

\item The decision/evolution tree of the above process is shown in Fig. \ref{fig:parent_hamiltonian}f. The computational time from $|Q_{t-1}\rangle $ to
$|Q_{t}\rangle $ is proportional to the average number of sub-steps

\begin{equation}
\eta _{t}+\sum_{k=0}^{\infty }(4k+6)(1-\eta _{t})^{2}\eta _{t}(\eta
_{t}^{2}+(1-\eta _{t})^{2})^{k}=\frac{1}{\eta _{t}}+1
\end{equation}
\end{description}

\item \textbf{Step 4:} After successful preparation of the state $|Q(z)\rangle $, we measure the operator $O$ (for inference) or $O_{1}$, $O_{2}$ (for training), and the expectation value of the measurement gives the required conditional probability or the gradient of KL divergence for learning.

\end{description}

The runtime of the whole algorithm described above is proportional to the
maximum of $1/\Delta _{t}$ and $1/\eta _{t}$. The gap $\Delta _{t}$ and the overlap $\eta _{t}$ depend on the topology of the graph $G$ and the parameters of the matrices $M_{i}$. If these two quantities are bounded by $\mbox{poly}(n)$ for all the steps $t$ from $1$ to $n$, this quantum algorithm will be efficient (runtime bounded by $\mbox{poly}(n)$). Although we do not expect this to be true in the worst case (even for the simplified classical generative model such as the restricted Boltzmann machine, the worst case complexity is at least NP hard \cite{dagum1993approximating}), we know that the QGM  with the above heuristic algorithm will provide exponential speedup over classical generative models for some instances. In the Supplementary Material \cite{sm}, we give a rigorous proof that our algorithm has exponential speedup over any classical algorithm for some instances under a reasonable conjecture. The major idea of this proof is as follows:

We construct a specific $|Q(z)\rangle $ which corresponds to the tensor network representation of the history state for universal quantum circuits rearranged into a two-dimensional (2D) spatial layout. The history state is a powerful tool in quantum complexity theory \cite{kitaev2002classical}. Similar type of history state has been used before to prove the QMA-hardness for spin systems in a 2D lattice \cite{oliveira2008complexity}. For this specific $|Q(z)\rangle $, we prove that both the gap $\Delta _{t}$ and the overlap $\eta _{t}$ scale as $1/\mbox{poly}(n)$ for all the steps $t$, by calculating the parent Hamiltonian of $|Q_{t}\rangle $ directly with proper grouping of
local tensors.  Our heuristic algorithm for training and inference therefore can be accomplished in polynomial time. On the other hand, our specific state $|Q(z)\rangle $ encodes universal quantum computation through representation of the history state, so it cannot be achieved by any classical algorithm in polynomial time if quantum computation cannot be efficiently simulated by a classical computer.  We summarize the result with the following theorem:

\begin{theorem}
\emph{There exist instances of computing conditional probability and
gradient of KL divergence to additive error $1/\mbox{poly}(n)$ such that (i)
our algorithm achieves it in polynomial time; (ii) any classical algorithm
cannot accomplish them in polynomial time unless universal quantum computing can be
simulated efficiently by a classical computer. }
\end{theorem}

In summary, we have introduced a quantum generative model for machine learning and proven that it offers exponential improvement in representational power over widely-used classical generative models. We have also proposed a heuristic quantum algorithm for training and making inference on our model, and proven that this quantum algorithm offers exponential speedup over any classical algorithm at least for some instances if quantum computing cannot be efficiently simulated by a classical computer. Our result combines the tools of different areas and generates an intriguing link between quantum many-body physics, quantum computation and complexity theory, and the machine learning frontier. This result opens a new route to apply the power of quantum computation to solving the challenging problems in machine learning and artificial intelligence, which, apart from its fundamental interest, has wide application potential.








\textbf{Acknowledgements} This work was supported by the Ministry of Education and the National key Research
and Development Program of China. LMD and ZYZ acknowledge in addition support from the MURI program.

\textbf{Author Information} All the authors contribute substantially to this work.
The authors declare no competing financial interests. Correspondence and requests
for materials should be addressed to L.M.D. (lmduanumich@gmail.com) or X. G. (gaoxungx@gmail.com).


\widetext

\clearpage
\setcounter{equation}{0}
\setcounter{figure}{0}
\setcounter{table}{0}
\makeatletter
\renewcommand{\theequation}{S\arabic{equation}}
\renewcommand{\thefigure}{S\arabic{figure}}
\renewcommand{\bibnumfmt}[1]{[S#1]}
\renewcommand{\citenumfont}[1]{S#1}

\begin{center}
    \textbf{\large Supplementary Material}
\end{center}

\bigskip

In this supplementary information, we provide rigorous proofs of theorem 2 and theorem 3 in the
main text. We also explain the relation between the representational power and the generational ability
of a generative model, the reduction of typical classical generative models to factor graphs, and more
details about the training and inference algorithms for our quantum generative model.

\section{ Representational power and generalization ability}

In this section, we briefly introduce some basic concepts in statistical
learning theory \cite{vapnik2013natures}, the theoretical foundation of
machine learning. Then we discuss the intuitive connection between the
generalization ability and representational power of a machine learning
model. This connection does not hold in the sense of mathematical rigor and counter examples exist, however, it is still a good guiding
principle in practice. More details can be found in the introductory book \cite%
{shalev2014understandings}.

A simplified way to formulate a machine learning task is as follows: 
given $M$ data generated independently from a distribution $%
\mathcal{D}$ (which is unknown), a set of hypothesis $\mathcal{H}=\{h\}$
(which could be a model with some parameters) and a loss function $L$
characterizing how ``good" a hypothesis $h$ is, try to find an $h_A$
produced by some machine learning algorithm $A$, minimizing the actual loss:
\begin{equation}
L_{\mathcal{D}}(h_A)\sim \epsilon_\text{bias}+\epsilon_\text{est}\text{
where }\epsilon_\text{bias}=\min_{h\in\mathcal{H}}L_{\mathcal{D}}(h) \text{
and } \epsilon_\text{est}\text{ depends on }M.
\end{equation}
The term $\epsilon_\text{bias}$ represents the bias error. When $\mathcal{H}$ is a hypothesis
class which is not rich enough or the number of parameters in the model is
small, this term might be large no matter how many training data are given,
which leads to underfitting. The term $\epsilon_\text{est}$ represents the generalization
error. When $\mathcal{H}$ is a very rich hypothesis class or the number of
parameters in the model is too large, this term might be large if the number
of training data is not enough, which leads to overfitting. This is the
bias-complexity trade-off of a machine learning model.

The generalization ability of a machine learning model is usually quantified by
sample complexity $M(\epsilon_\text{est},\delta)$ in the framework of
Probably Approximately Correct (PAC) learning \cite{valiant1984theorys}, which
means if the number of training data is larger than $M(\epsilon_\text{est}%
,\delta)$, the generalization error is bounded by $\epsilon_\text{est}$ with
probability at least $1-\delta$. It has been proved that the sample complexity has
the same order of magnitude as a quantity of the hypothesis set $\mathcal{H}
$, the Vapnik-Chervonenkis dimension (VC dimension)\cite{vapnik2015uniforms}.
The VC dimension could be used to characterize the ``effective" degree of
freedom of a model, which is the number of parameters in most situations.
With a smaller VC dimension, the generalization ability gets better.

There are some caveats of using the VC dimension to connect the
generalization ability and the representational power of a machine learning
model though, which include: (i) the VC dimension theory only works in the framework of PAC
learning (thus restricted to supervised learning), so it is not directly
applicable to generative models; (ii) the number of parameters does not always
match the VC dimension, e.g., the combination of several parameters as $%
\theta_1+\theta_2+\cdots$ only counts as one effective parameter, meanwhile
there also exists such a model with only one parameter but having infinite VC dimensions. 
However, despite of existence of those counter examples, the VC dimension matches the number of
parameters in most situations, so it is still a guiding principle to choose models with a smaller number of parameters. If the
number of parameters is large, in order to determine each parameter, it
needs a large amount of information, thus a large number of data is
necessary.

In the case of the QGM, we find a distribution (illustrated by a blue circle in Fig. \ref%
{fig:FGvsQGM}) such that, in order to represent it or equivalently make $%
\epsilon_\text{bias}$ small, the factor graph should have a number
of parameters exponentially large (Fig. \ref{fig:FGvsQGM}b). In this case, the distribution that
the factor graph could represent has a very large degrees of freedom.
However, the distribution which the QGM could represent only occupies a small
corner of the whole distribution space (Fig. \ref{fig:FGvsQGM}b), so the
information needed to determine the parameters will be much smaller. Similar
reasoning has been used to illustrate the necessity of using deep
architecture in neural network \cite%
{goodfellow2016deeps,maass1994comparisons,maass1997boundss,montufar2014numbers,poggio2015theorys,eldan2015powers,mhaskar2016learning1s,raghu2016expressives,poole2016exponentials,mhaskar2016learning2s,lin2016doess,liang2016deeps}%
.
\begin{figure}[tbp]
\centering
\includegraphics[width=0.5\linewidth]{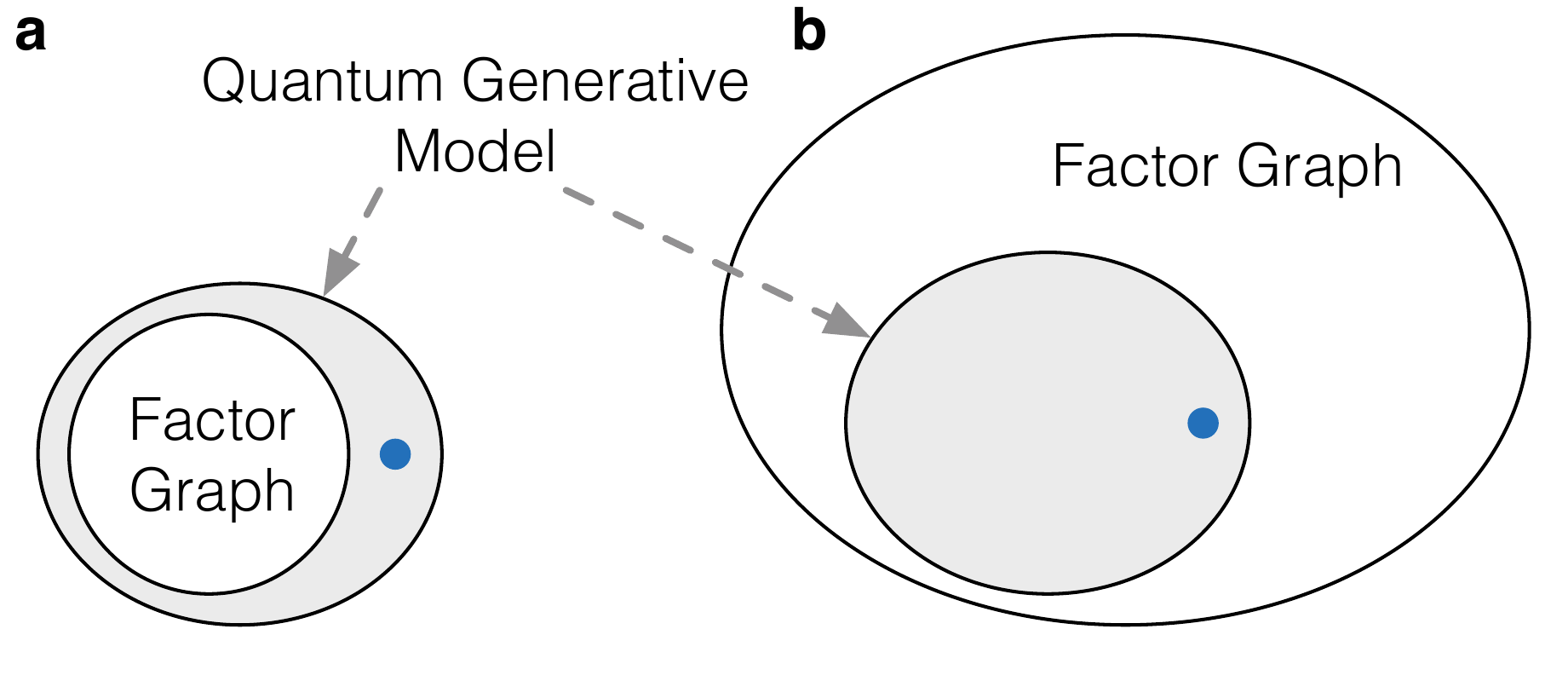}
\caption{ \textbf{Parameter Space of Factor Graph and QGM. a,} The case when both models
have a polynomial number of parameters. In this case, factor graphs cannot represent
some distributions in the QGM illustrated as blue circles. \textbf{b,} In order to
represent the blue circle distribution from the QGM, factor graphs have to
involve an exponentially large number of parameters. In this case, the parameter space will
inflate to a very large scale. }
\label{fig:FGvsQGM}
\end{figure}

\section{Reduction of typical generative models to factor graphs}

\label{sec:GNN}

In this section, we review typical classical generative models with graph structure. There are
two large classes: one
is the probabilistic graphical model \cite%
{koller2009probabilistics,bishop2006patterns}, and the other is the energy-based
generative neural network. Strictly speaking, the later one is a special
case of the former, but we regard it as another category because it is usually
referred in the context of deep learning \cite{goodfellow2016deeps}.
We discuss the ``canonical form" of generative models, the factor
graph, and show how to convert various classical generative models into this canonical form. The details could be found in the books in Refs. \cite%
{koller2009probabilistics,bishop2006patterns}.

\begin{figure}[tbp]
\centering
\includegraphics[width=0.5\textwidth]{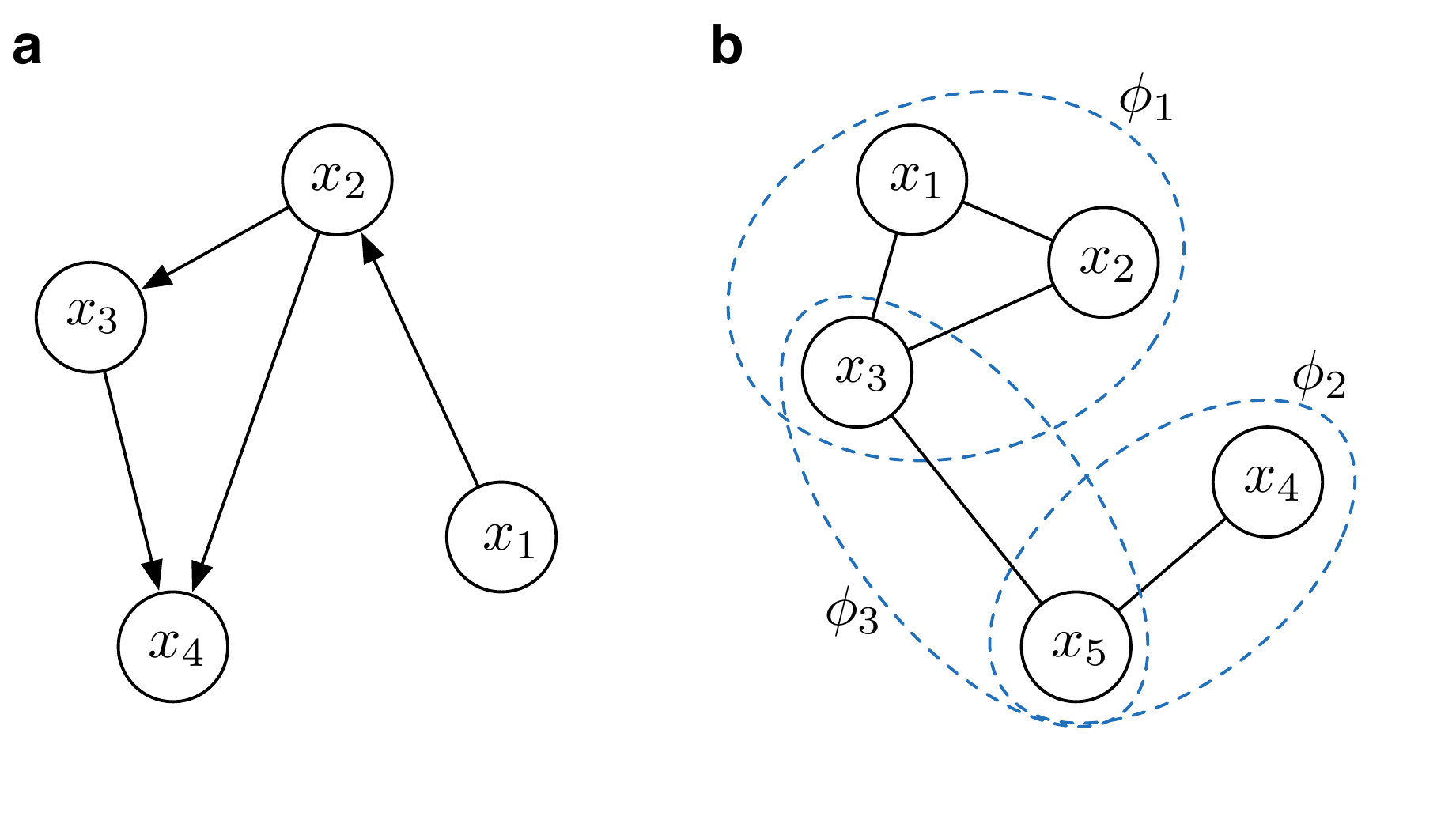}
\caption{ \textbf{Probabilistic Graphical Models. a,} The Bayesian network. The
distribution defined by this figure is $%
P(x_1,x_2,x_3,x_4)=P(x_1)P(x_2|x_1)P(x_3|x_2)P(x_4|x_2,x_3)$. \textbf{b, }%
Markov random field. The distribution defined by this figure is $%
P(x_1,x_2,x_3,x_4)=\protect\phi_1(x_1,x_2,x_3)\protect\phi_2(x_4,x_5)\protect%
\phi_3(x_3,x_5)/Z$ where $Z=\sum_{x_1,x_2,x_3,x_4}P(x_1,x_2,x_3,x_4)$ is the
normalization factor or partition function. Each dashed circle
corresponds to a clique. }
\label{fig:PGM}
\end{figure}


First, we introduce two probabilistic graphical models as shown in
Fig. \ref{fig:PGM}: the Bayesian network (directed graphical model) and the Markov
random field (undirected graphical model). A Bayesian network (Fig. \ref%
{fig:PGM}a) is defined on a directed acyclic graph. For each node $x_i$,
assign a transition probability $P(x_i|\text{parents of }x_i)$ where the parent
means starting point of a directed edge of which the end point is $x_i$. If
there is no parent for $x_i$, assign a probability $P(x_i)$. The total
probability distribution is the product of these conditional probabilities.
A Markov random field (Fig. \ref{fig:PGM}b) is defined on an undirected graph.
The sub-graph in each dashed circle is a clique in which each pair of nodes
is connected. For each clique, assign a non-negative function for each node
variable. The total probability distribution is proportional to the product
of these functions.

\begin{figure}[h]
\centering
\includegraphics[width=1\textwidth]{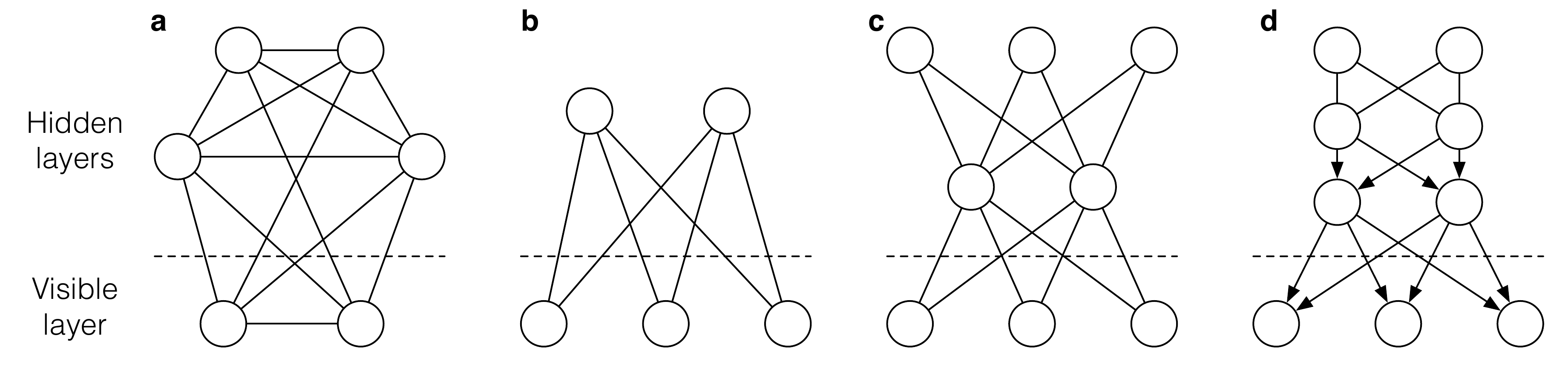}
\caption{ \textbf{Energy-based Neural Networks.} \textbf{a,} Boltzmann
machine. A fully connected graph. The correlation between each pair $x_i$
and $x_j$ is $e^{ax_ix_j+bx_i+cx_j+d}$ with $a,b,c,d$ being real numbers. The
whole distribution is proportional to the product of all the correlation
between each pair of nodes, which is a Gibbs distribution with two-body
interaction and classical Hamiltonians. \textbf{b,} Restricted Boltzmann
Machine. Bipartite graph (graph with only two layers). The correlation is
the same as in \textbf{a} except $a=b=c=d=0$ between the pairs not connected. 
There is no connection between pairs in the same layer.
\textbf{c,} Deep Boltzmann machine. Basically the same as \textbf{b} except
there are more than two hidden layers. \textbf{d,} Deep belief network. A
mixture of undirected and directed graphical model. The top two
layers define a restricted Boltzmann machine. For each node $%
y_j$ as the end point of directed edges, it is assigned a conditional
probability $P(y_j=0|\cdots,x_i,\cdots)=1/(1+e^{\sum_i a_{ij} x_i+b_j})$. The
distribution of the variables in the visible layer is defined as the marginal
probability. }
\label{fig:GNN}
\end{figure}

Then we introduce four typical generative neural networks as shown in Fig. \ref{fig:GNN}:
the Boltzmann machine, the restricted Boltzmann machine, the deep Boltzmann machine, and 
the deep belief network. These neural networks are widely used in deep
learning \cite{goodfellow2016deeps}.

All the above generative models could be represented in the form of factor
graph. Factor graph could be viewed as a canonical form of the generative
models with graph structure. The factor graph representation of the Bayesian
network and the Markov random field simply regards the conditional probability $%
P(x|y)$ and function $\phi(x,y)$ as the factor correlation $f(x,y)$. The factor graph
representations of the Boltzmann machine, the restricted Boltzmann machine and the deep
Boltzmann machine are graphs such that there is a square in each edge
 in the original network with correlation $%
f(x_i,x_j)=e^{ax_ix_j+bx_i+cx_j}$. The factor graph with unbounded degrees
in these cases could be simulated by a factor graph with a constantly bounded degree, which is shown in Fig. \ref{fig:local}. The idea is to add equality
constraints which could be simulated by correlation $%
e^{ax_ix_j-ax_i/2-ax_j/2}$ with $a$ being a very large positive number since
\begin{equation}
e^{ax_ix_j-ax_i/2-ax_j/2}=\delta_{x_ix_j}+e^{-a/2}(1-\delta_{x_ix_j})%
\longrightarrow \delta_{x_ix_j} \text{ as } a\rightarrow+\infty.
\end{equation}
The factor graph representation of the deep belief network is a mixture of
the Bayesian network and the restricted Boltzmann machine except that one correlation in the
part of directed graph involves an unbounded number of variables. This is the conditional probability with one variable conditioned on
the values of variables in the previous layer. It is
\begin{equation}
P(y|\cdots,x_i,\cdots)=\frac{e^{(\sum_ia_{ij} x_i+b_j)y}}{1+e^{\sum_ia_{ij}
x_i+b_j}}.
\end{equation}
where $\{\cdots,x_i,\cdots\}$ are the parents of $y$. Given the values of $%
\cdots,x_i,\cdots$, it's easy to sample $y$ according to $%
P(y|\cdots,x_i,\cdots)$ (which is actually the original motivation for introducing the deep belief network). Thus the process could be represented by a Boolean circuit with random coins. In fact, a circuit with random coins is a
Bayesian network: each logical gate is basically a conditional probability
(for example, the AND gate $y=x_1\land x_2$ could be simulated by the
conditional probability $P(y|x_1x_2)=\delta_{y,x_1\cdot x_2}$). So this
conditional probability could be represented by a Bayesian network in which
the degree of each node is bounded by a constant. Thus deep belief network
could be represented by a factor graph with a constant degree.
\begin{figure}[tbp]
\centering
\includegraphics[width=0.5\textwidth]{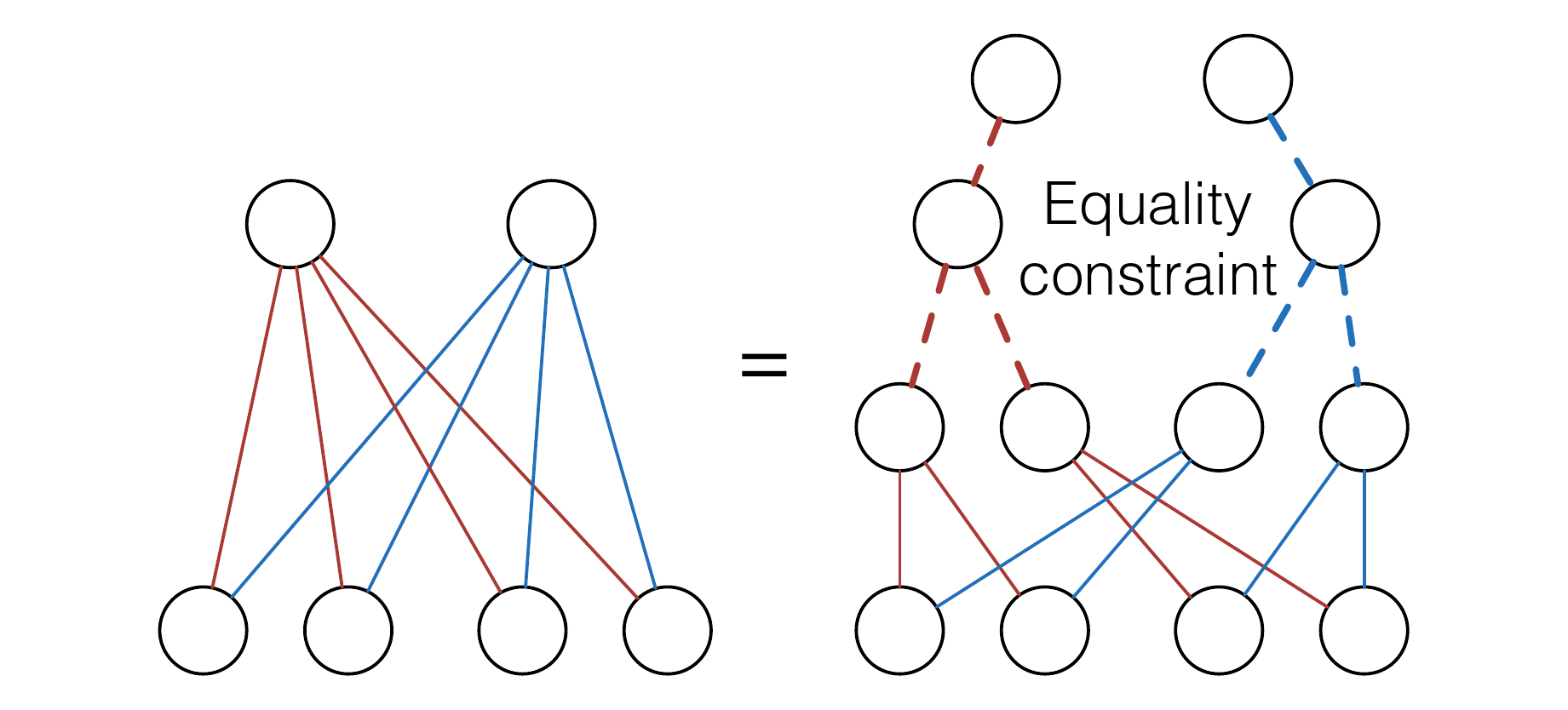}
\caption{ \textbf{Simulating graphs of unbounded degrees with
graphs of constantly bounded degrees.} In the case that all the correlations involve only two variables (e.g.,
for Boltzmann machines), the correlation between one node and $l$ other nodes
could be simulated by a binary tree structure with depth $\log l$. The newly
added nodes and connections are marked by the dashed circles and lines,
respectively. The correlation in the dashed line is an equality constraint
which could be approximated by $e^{ax_ix_j-ax_i/2-ax_j/2}$ with $a$ being a
very large positive number.  }
\label{fig:local}
\end{figure}

One important property is that the conditional probability of a factor graph
is still factor graph. Actually, the correlation becomes simpler: the number
of variables involved does not increase, which means approximately computing
the conditional probability of probabilistic graphical models could be
reduced to preparing the state $|Q\rangle$ (which will be shown in next section 
). Then we arrive at the following lemma:

\begin{lemma}
\label{lem:FG}\emph{\ All the probability distributions or conditional
probability distributions of probabilistic graphical models and energy-based
neural networks can be represented efficiently by the factor graph with a
constant degree.}
\end{lemma}

\section{Parent Hamiltonian of the State $|Q\rangle $}

We consider tensor network representation of $|Q\rangle $ defined on a graph with a constant
degree.
\begin{equation}
|Q\rangle \equiv M_{1}\otimes \cdots \otimes M_{n}|G\rangle
\end{equation}%
where $|G\rangle $ is the graph state that is the unique ground state of the
frustration-free local Hamiltonian with zero ground-state energy:
\begin{equation}
H_{G}=\sum_{i}H_{i}=\sum_{i}\frac{I-\left( \bigotimes_{j\in \{i\text{'s
neighbors}\}}Z_{j}\right) \otimes X_{i}}{2}
\end{equation}%
where each $H_{i}$ is the projector of the stabilizer \cite%
{gottesman1997stabilizers,raussendorf2001ones} supported only on the node $i$
and its neighbors, thus being local since the degree of the graph is a constant.
And
\begin{equation}
H_{i}|G\rangle =0.
\end{equation}

Next we construct a Hamiltonian:
\begin{equation}
H_{Q}=\sum_{i}H_{i}^{\prime }=\sum_{i}\left( \bigotimes_{j\in \{i\}\cup \{i%
\text{'s neighbors}\}}M_{j}^{-1}\right) ^{\dag }H_{i}\left( \bigotimes_{j\in
\{i\}\cup \{i\text{'s neighbors}\}}M_{j}^{-1}\right) .
\end{equation}%
First, we show that $|Q\rangle $ is the ground state of $H_{Q}$. Since $H_{i}
$ is positive semidefinite, thus $H_{i}^{\prime }$ is also positive
semidefinite, the eigenvalue of $H_{Q}$ is no less than $0$.
\begin{eqnarray}
\langle Q|H_{i}^{\prime }|Q\rangle  &=&\langle Q|\left( \bigotimes_{j\in
\{i\}\cup \{i\text{'s neighbors}\}}M_{j}^{-1}\right) ^{\dag }H_{i}\left(
\bigotimes_{j\in \{i\}\cup \{i\text{'s neighbors}\}}M_{j}^{-1}\right)
|Q\rangle   \notag \\
&=&\langle G|\left( \bigotimes_{k\in \text{all the nodes except for }j\text{s%
}}(M_{k}M_{k}^{\dag })^{-1}\right) \otimes H_{i}|G\rangle   \notag \\
&=&0,
\end{eqnarray}%
so $\langle Q|H_{Q}|Q\rangle =0$ which means $|Q\rangle $ is the ground
state of $H_{Q}$. Then we prove $|Q\rangle $ is the unique ground state.
Suppose $|Q^{\prime }\rangle =M_{1}\otimes \cdots \otimes M_{n}|G^{\prime
}\rangle $ satisfying $\langle Q^{\prime }|H_{Q}|Q^{\prime }\rangle =0$ for
some state $|G^{\prime }\rangle $, which implies $\langle Q^{\prime
}|H_{i}^{\prime }|Q^{\prime }\rangle =0$ for every $i$. So
\begin{eqnarray}
\langle Q^{\prime }|H_{i}^{\prime }|Q^{\prime }\rangle  &=&\langle Q^{\prime
}|\left( \bigotimes_{j\in \{i\}\cup \{i\text{'s neighbors}%
\}}M_{j}^{-1}\right) ^{\dag }H_{i}\left( \bigotimes_{j\in \{i\}\cup \{i\text{%
's neighbors}\}}M_{j}^{-1}\right) |Q^{\prime }\rangle   \notag \\
&=&\langle G^{\prime }|H_{i}\otimes \left( \bigotimes_{k\in \text{all the
nodes except for }j\text{s}}(M_{k}M_{k}^{\dag })^{-1}\right) |G^{\prime
}\rangle   \notag \\
&=&\left\vert H_{i}\otimes \left( \bigotimes_{k\in \text{all the nodes
except for }j\text{s}}M_{k}^{-1}\right) |G^{\prime }\rangle \right\vert _{2}
\notag \\
&=&0,
\end{eqnarray}%
which implies
\begin{equation}
H_{i}\otimes \left( \bigotimes_{k\in \text{all the nodes except for }i\text{
and }j\text{s}}M_{k}^{-1}\right) |G^{\prime }\rangle =0\Longrightarrow
H_{i}|G^{\prime }\rangle =0
\end{equation}%
for every $i$, thus $|G\rangle =|G^{\prime }\rangle $. This proves the
uniqueness of $|Q\rangle $ as the ground state of $H_{Q}$.

\section{Proof of Theorem 2}

In this section, we use computational complexity theory to prove the
exponential improvement on representational power of the QGM over any
factor graph in which each correlation is easy to compute given a specific
value for all variables (including all the models mentioned above). The 
discussion of the relevant computational complexity theory could be
found in the book in Ref. \cite{arora2009computationals} or in the recent  review article 
on quantum supremacy \cite{Harrow2017s}.



To prove theorem 2, first we define the concept of multiplicative error. Denote the probability distribution produced by the QGM as $\{q(x)\}$. Then we
ask whether there exists a factor graph such that its distribution $\{p(x)\}$
approximates $\{q(x)\}$ to some error. It is natural to require that if $q(x)$ is very small, $p(x)$ should also be very small, which means rare events should still be rare. So we define the
following error model:
\begin{definition}[Multiplicative Error]
\emph{\ Distribution $\{p(x)\}$ approximates distribution $\{q(x)\}$ to
multiplicative error means
\begin{equation}
|p(x)-q(x)|\le\gamma q(x)
\end{equation}
for any $x$, where $\gamma=\Omega(1/\mbox{poly}(n))<1/2$.}
\end{definition}
\noindent This error can also guarantee that any local behavior is approximately the
same since the $l_{1}$-distance can be bounded by it:
\begin{equation}
\sum_{x}|p(x)-q(x)|\leq \sum_{x}\gamma q(x)=\gamma .
\end{equation}%
But only bounding $l_{1}$-distance cannot guarantee that rare event is still rare.
Multiplicative error for small $\gamma $ implies that the KL-divergence is
bounded by
\begin{eqnarray}
\left\vert \sum_{x}q(x)\log \frac{p(x)}{q(x)}\right\vert  &\leq
&\sum_{x}q(x)\log \left( 1+\left\vert \frac{p(x)}{q(x)}-1\right\vert \right)
\notag \\
&\leq &\sum_{x}q(x)\gamma =\gamma .
\end{eqnarray}%

The probability $p(x|z)$ of a factor graph can be written as
\begin{equation}
p(x|z)=\frac{\sum_yp(x,y,z)}{\sum_{x,y}p(x,y,z)}=\frac{\sum_yf(x,y)}{%
\sum_{x,y}f(x,y)}
\end{equation}
where $f(x,y)$ is product of a polynomial number of relatively simple
correlations, thus non-negative and can be computed in polynomial time and $y$
are hidden variables.

The probability $q(x)$ of the QGM can be written as
\begin{equation}
q(x)=\frac{\sum_yg(x,y)}{\sum_{x,y}g(x,y)}
\end{equation}
where $g(x,y)$ is product of a polynomial number of tensors given a specific
assignment of indices and $y$ are the virtual indices or physical indices of
hidden variables, $x$ is the remaining physical indices. Different from $%
f(x,y)$, $g(x,y)$ is a complex number in general. Thus in some sense, we can
say $q(x)$ is the result of quantum interference in contrast to the case of $%
p(x)$ (or $p(x|z)$) which is only summation of non-negative numbers. Since $%
q(x)$ is summation of complex numbers, in the process of summation, it can
oscillate dramatically, so we expect $q(x)$ being more complex than $p(x)$ (or $%
p(x|z)$). This can be formalized as the following lemma.

\begin{lemma}[Stockmeyer's theorem \protect\cite{stockmeyer1976polynomials}]
\emph{\ There exists an $\mathsf{FBPP}^\mathsf{NP}$ algorithm which can
approximate
\begin{equation}
P=\Pr_t[f(t)=1]=\frac{1}{2^r}\sum_{t\in\{0,1\}^r}f(t)
\end{equation}
by $\widetilde P$, for any boolean function $f:\{0,1\}^r\rightarrow \mathbb{R%
}^+\cup \{0\}$, to multiplicative error $|\widetilde P-P|\le P/\mbox{poly}(n)
$ if $f(t)$ can be computed efficiently given $t$.}
\end{lemma}
\noindent
$\mathsf{FBPP}^\mathsf{NP}$ algorithms denote algorithms that a probabilistic
Turing machine, supplied with an oracle that can solve all the NP problems
in one step, can run in polynomial time.

Without the constraint $f(t)\ge0$, approximating $P$ by $\widetilde P$ such
that $|P-\widetilde P|\le\gamma P$ is in general \textsf{\#P}-hard even if $%
\gamma<1/2$. Roughly speaking, the complexity class \textsf{\#P} \cite 
{valiant1979complexitys} includes problems counting the number of
witnesses of an NP problem, which is believed much harder than NP (see Ref.  
\cite{Harrow2017s} for a quantum computing oriented introduction). The above
lemma shows that summation of non-negative numbers is easier than complex
numbers in general. This formulates that quantum interference is more
complex than classical probability superposition.

Stockmeyer's theorem was firstly used to separate classical and quantum
distribution in Ref. \cite{aaronson2011computationals} where the classical
distribution is sampled by a probabilistic Turing machine in polynomial time
and the quantum distribution is sampled from linear optics network. In some
sense, our result could be viewed as a development of this result: the
classical distribution is not necessarily sampled by a classical device
efficiently, instead, the distribution could be approximated by a factor graph
to multiplicative error. An efficient classical device is a special case of
factor graph because it could be represented as a Boolean circuit with random coins, this could be represented as 
the Bayesian network as we have discussed above about the representation of deep belief
network by factor graph. 

Then we give the proof of theorem 2:
\begin{proof}
Assume there exists a factor graph from which we can compute the conditional probability $p(x|z)$,
we will prove that approximating $q(x)$ to multiplicative error, i.e., $|p(x|z)-q(x)|\le\gamma q(x)$, is in $\mathsf{FBPP}^\mathsf{NP}/\poly$ (the meaning of this complexity class will be explained later).

Suppose $f$ is defined as $p(x,y|z)$, there exists an $\mathsf{FBPP}^\mathsf{NP}$ algorithm approximating $P_1=\sum_{x,y}f(x,y)$ by $\widetilde P_1$ such that $|P_1-\widetilde P_1|\le\gamma_1P_1$ and $P_2=\sum_{y}f(x,y)$ by $\widetilde P_2$ such that $|P_2-\widetilde P_2|\le\gamma_2P_2$. Define $\widetilde p(x)=\widetilde P_2/\widetilde P_1$ and we have $p(x|z)=P_2/P_1$
\begin{eqnarray}
\nonumber|\widetilde p(x)-p(x|z)|&=&\left| \frac{\widetilde P_2}{\widetilde P_1}-\frac{ P_2}{ P_1}  \right|\\
\nonumber&\le&\left| \frac{\widetilde P_2}{\widetilde P_1}-\frac{ P_2}{\widetilde P_1}  \right|+\left| \frac{ P_2}{\widetilde P_1}- \frac{ P_2}{P_1} \right|\\
\nonumber&=&\frac{|\widetilde P_2-P_2|}{\widetilde P_1}+P_2\left| \frac{1}{\widetilde P_1}-\frac{1}{P_1}  \right|\\
\nonumber&=&\frac{|\widetilde P_2-P_2|}{\widetilde P_1}+\frac{P_2}{\widetilde P_1 P_1}|\widetilde P_1-P_1|\\
\nonumber&\le&(\gamma_1+\gamma_2)\frac{P_2}{\widetilde P_1}\\
\nonumber&\le& \frac{\gamma_1+\gamma_2}{1-\gamma_1}\frac{P_2}{P_1}\\
&=&\frac{\gamma_1+\gamma_2}{1-\gamma_1} p(x|z),
\end{eqnarray}
then
\begin{equation}
|\widetilde p(x)-q(x)|\le|\widetilde p(x)-p(x|z)|+|p(x|z)-q(x)|\le\frac{\gamma_1+\gamma_2}{1-\gamma_1} p(x|z)+\gamma q(x)\le\frac{\gamma+\gamma_1+\gamma_2+\gamma_2\gamma}{1-\gamma_1}q(x)<\frac{1}{2}q(x),
\end{equation}
the last step is because we choose $\gamma_1$ and $\gamma_2$ as sufficiently small as $1/\poly(n)$. Under the assumption that the representation is efficient, such $f(x,y)$ can be represented by a polynomial size circuit, where the circuit corresponds to the description of the factor graph. So $\widetilde p(x)$ can be computed in $\mathsf{FBPP}^\mathsf{NP}/\poly$. ``/\poly" is because we do not need to construct the circuit efficiently \cite{karp1982turings}.

\begin{figure}
\centering
\includegraphics[width=0.5\linewidth]{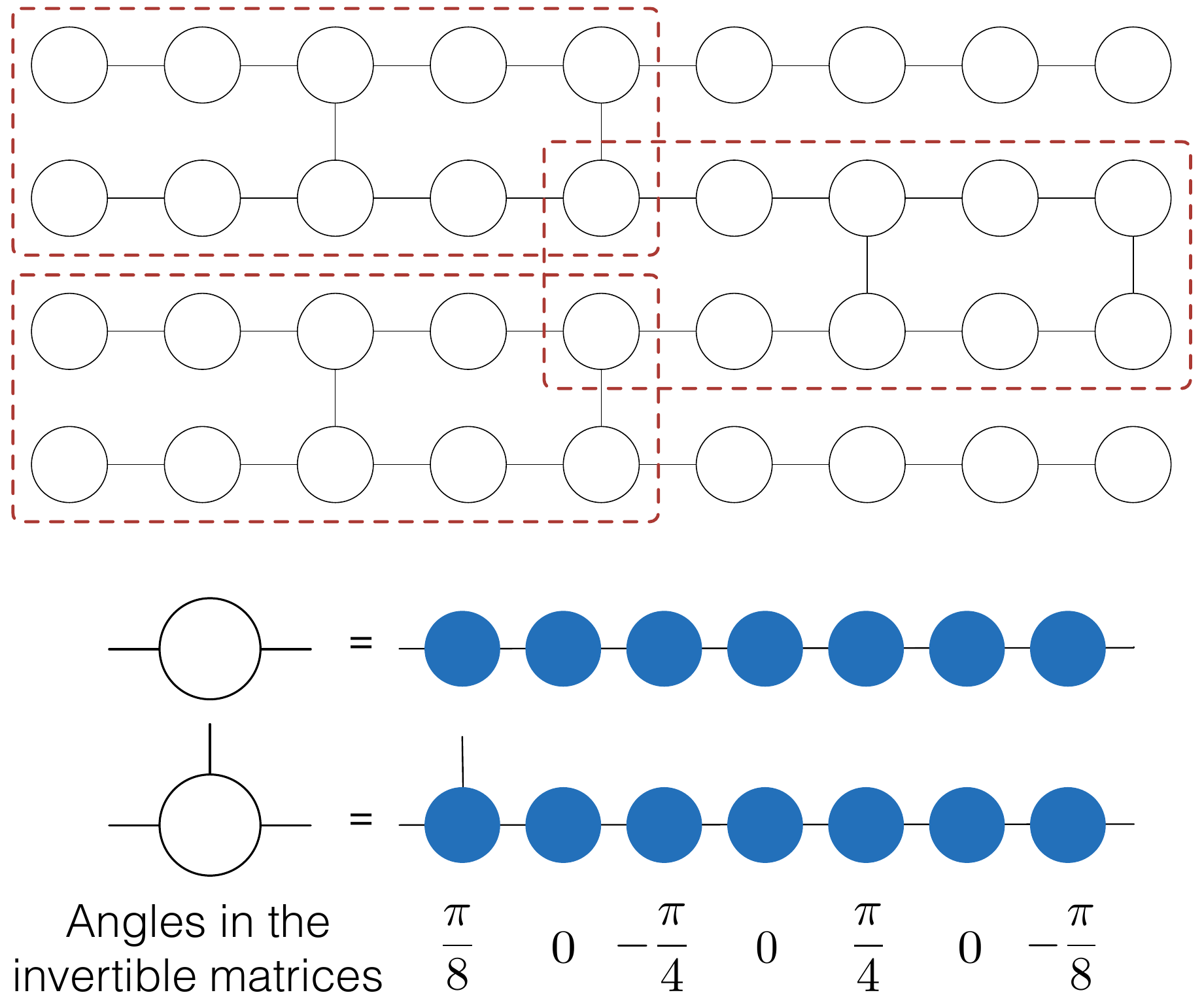}\caption{\textbf{\#P-hardness for QGM.} The state $\ket Q$ in Ref. \cite{gao2016quantums}.
To construct this state, we start from a brickwork of white circles \cite{broadbent2009universals} (the top side), with each white circle
representing seven blue circles. Each blue circle represents a qubit. Then we apply $HZ(\theta)$ which is clearly an invertible matrix on each qubit with the angle $\theta$ shown at the bottom.}%
\label{fig:qs}%
\end{figure}

In Ref. \cite{gao2016quantums}, we introduced a special form of QGM that corresponds to a graph state $\ket Q$ (Fig. \ref{fig:qs}) with one layer of invertible matrices $HZ(\theta)$ such that computing $q(x)$ to multiplicative error with $\gamma<1/2$ is at least $\mathsf{FP}^\mathsf{\#P}$. So assuming the efficient representation of QGM by factor graph, we will get
\begin{equation}
\mathsf{P}^\mathsf{\#P}\subseteq  \mathsf{BPP}^\mathsf{NP}/\poly.
\end{equation}
Then it follows basically the same reasoning as the proof of theorem 3 in Ref. \cite{aaronson2017implausibilitys} except they consider $\mathsf{P}^\mathsf{\#P}\subseteq\mathsf{BPP}^{\mathsf{NP}^\mathsf{NP}}/\poly$ and the result is polynomial hierarchy collapse to the fourth level.
According to Toda's theorem \cite{toda1989computationals}, \textsf{PH}$\subseteq\mathsf{P}^\mathsf{\#P}$, this implies $\mathsf{NP}^\mathsf{NP}\subseteq\mathsf{BPP}^\mathsf{NP}/\poly$. According to Adleman's result \cite{adleman1978twos}, \textsf{BPP}$\subseteq$\textsf{P}/\poly, relativizes, which means $\mathsf{NP}^\mathsf{NP}\subseteq \mathsf{P}^\mathsf{NP}/\poly$. Karp-Lipton theorem \cite{karp1982turings} states if $\textsf{NP}\subset \textsf{P}/\poly$, then $\Sigma^p_2\subseteq\Pi^p_2$ (polynomial hierarchy collapse to the second level); the result is also relativizing then it follows if $\mathsf{NP}^\mathsf{NP}\subseteq \mathsf{P}^\mathsf{NP}/\poly$, then ${\Sigma^p_2}^\textsf{NP}\subseteq{\Pi^p_2}^\textsf{NP}$ which means $\Sigma^p_3\subseteq\Pi^p_3$ (polynomial hierarchy collapse to the third level).
\end{proof}

\section{Training and inference in quantum generative model}

\label{sec:computation_on_our_model} In this section, we discuss how to train the QGM and make inference on it. First, we briefly review how
to train and make inference on some typical factor graphs.
Then we reduce inference on the QGM to preparation of a tensor network
state. Finally, we derive the formula for computing gradient of the KL-divergence
of QGM and reduce it to the preparation of a tensor network state.

Inference problems on probabilistic graphical models include
computing marginal probability $\sum_yp(x,y)$ and conditional probability $%
\sum_yp(x,y|z)$ (which includes marginal probability as a special case when the set $z$ is empty). To approximately compute the probability on
some variables, we sample the marginal probability $p(x,y)$ or the conditional probability $p(x,y|z)$, and then measure the
variable set $x$. We use the Boltzmann machine as an example to show how to train energy-based neural networks. The KL-divergence given $M$ data is
\begin{equation}
\frac{1}{M}\sum_{v\in \text{data set}}\log p(v).
\end{equation}
The training is to optimize this quantity with the gradient descent
method. The gradient (for simplicity, we only present $\partial_a$ between
hidden and visible nodes) is
\begin{eqnarray}
\frac{1}{M} \partial_{a_{ij}}\sum_{v\in \text{data set}}\log p(v)&=&\frac{1}{M}\sum_{v\in \text{data set}}\sum_hP(h|v)h_iv_j-\sum_{h,v}P(h,v)h_iv_j
\notag \\
&=&\langle h_iv_j\rangle_\text{data}-\langle h_iv_j\rangle_\text{model}.
\end{eqnarray}
The subscript ``data" denotes distribution with probability $1/M$ randomly chosen from
the training data according to $P(h|v)=P(h,v)/P(v)$, where $P(h,v)$ is the
distribution defined by the graphical model, randomly sampling hidden
variables $h$. The distribution ``model" is $P(h,v)$. So the training is
reduced to sampling some distribution or conditional distribution defined by the generative neural network and then estimating the
expectation value of local observables. Since the QGM could represent
conditional probability of these models and the corresponding state $%
|Q\rangle$ is the unique ground state of local Hamiltonian, inference and training could be reduced to the ground
state preparation problem.

Similarly, inference on the QGM could also be reduced to preparation of a
quantum state. As an example, let us compute the marginal probability for the QGM:
\begin{equation}
q(x)=\sum_yq(x,y)=\sum_y \frac{\langle Q|x,y\rangle\langle x,y|Q\rangle}{\innerp{Q}{Q}}=\frac{\langle
Q|(|x\rangle\langle x|\otimes I)|Q\rangle}{\innerp{Q}{Q}}=\frac{\langle Q|(O\otimes I)|Q\rangle}{\innerp{Q}{Q}}.
\end{equation}
So the problem is reduced to preparing the state $|Q\rangle$ and then measuring
the local observable $O$. Similarly, the conditional probability is given by
\begin{equation}
q(x|z)=\sum_y\frac{q(x,y,z)}{q(z)}=\frac{\langle Q(z)|(|x\rangle\langle x|\otimes
I)|Q(z)\rangle}{\innerp{Q(z)}{Q(z)}}=\frac{\langle Q(z)|(O\otimes I)|Q(z)\rangle}{\innerp{Q(z)}{Q(z)}},
\end{equation}
where
\begin{equation}
|Q(z)\rangle\equiv (I\otimes \langle z|)|Q\rangle
\end{equation}
is a tensor network state.

The KL-divergence of the QGM is given by
\begin{equation}
\frac{1}{M}\sum_{v\in\text{data set}}\log{\langle Q(v)|Q(v)\rangle}%
-\log\langle Q|Q\rangle,
\end{equation}
and its derivative with respect to a parameter $\theta_i$ in $M_i$ is
\begin{eqnarray}
\partial_{\theta_i}\left(\frac{1}{M}\sum_{v\in\text{data set}}\log{\langle
Q(v)|Q(v)\rangle}-\log\langle Q|Q\rangle \right)&=&\frac{1}{M} \sum_{v\in%
\text{data set}}\frac{\partial_{\theta_i}\langle Q(v)|Q(v)\rangle}{\langle
Q(v)|Q(v)\rangle}-\frac{\partial_{\theta_i}\langle Q|Q\rangle}{\langle
Q|Q\rangle}.
\end{eqnarray}
Let us consider the second term first.
\begin{eqnarray}
\frac{\partial_{\theta_i}\langle Q|Q\rangle}{\langle Q|Q\rangle}&=&\frac{%
\partial_{\theta_i} \langle G | M_1^\dag M_1\otimes \cdots \otimes M_n^\dag
M_n|G \rangle }{\langle Q|Q\rangle}  \notag \\
&=&\frac{\langle G | M_1^\dag M_1\otimes \cdots
\otimes\partial_{\theta_i}M_i^\dag M_i\otimes \cdots \otimes M_n^\dag M_n|G
\rangle }{\langle Q|Q\rangle}  \notag \\
&=&\frac{\langle G | M_1^\dag M_1\otimes \cdots \otimes \left(M_i^\dag
(\partial_{\theta_i}M_i)+(\partial_{\theta_i}M_i^\dag)M_i\right) \otimes
\cdots \otimes M_n^\dag M_n|G \rangle }{\langle Q|Q\rangle}  \notag \\
&=&\frac{\langle G | M_1^\dag M_1\otimes \cdots \otimes \left(M_i^\dag
(\partial_{\theta_i}M_i)M_i^{-1}M_i+M_i^\dag
M_i^{\dag-1}(\partial_{\theta_i}M_i^\dag)M_i\right) \otimes \cdots \otimes
M_n^\dag M_n|G \rangle }{\langle Q|Q\rangle}  \notag \\
&=&\frac{\langle Q | (\partial_{\theta_i}M_i)M_i^{-1}+
M_i^{\dag-1}(\partial_{\theta_i}M_i^\dag) |Q\rangle}{\langle Q|Q\rangle}.
\end{eqnarray}
It is basically the same for $|Q(v)\rangle$ if $\theta_i$ is the parameter
for the unconditioned qubit:
\begin{equation}
\frac{\partial_{\theta_i}\langle Q(v)|Q(v)\rangle}{\langle Q(v)|Q(v)\rangle}=%
\frac{\langle Q(v) |(\partial_{\theta_i}M_i)M_i^{-1}+
M_i^{\dag-1}(\partial_{\theta_i}M_i^\dag) |Q(v)\rangle}{\langle
Q(v)|Q(v)\rangle}.
\end{equation}
If $\theta_i$ is the parameter of conditioned variables, it is more
complicated. Let $|R_i\rangle=M_i^{-1}(I\otimes \langle
v/v_i|)|Q(v)\rangle$ which is independent of $\theta_i$ and $(I\otimes
\langle v_i|M_i)|R_i\rangle=|Q(v)\rangle$, $(I\otimes
M_i)|R_i\rangle=|Q(v/v_i)\rangle$, where $v/v_i$ denote
all variables in $v$ except $v_i$.
\begin{eqnarray}
\frac{\partial_{\theta_i}\langle Q(v)|Q(v)\rangle}{\langle Q(v)|Q(v)\rangle}%
&=&\frac{\langle R_i| \partial_{\theta_i}(M_i^\dag|v_i\rangle\langle
v_i|M_i) |R_i\rangle}{\langle Q(v)|Q(v)\rangle}  \notag \\
&=& \frac{\langle R_i| (\partial_{\theta_i}M_i^\dag)|v_i\rangle\langle
v_i|M_i+M_i^\dag|v_i\rangle\langle v_i|(\partial_{\theta_i}M_i) |R_i\rangle}{%
\langle Q(v)|Q(v)\rangle}  \notag \\
&=& \frac{\langle R_i| M_i^\dag
M_i^{\dag-1}(\partial_{\theta_i}M_i^\dag)|v_i\rangle\langle
v_i|M_i+M_i^\dag|v_i\rangle\langle v_i|(\partial_{\theta_i}M_i)M_i^{-1}M_i
|R_i\rangle}{\langle Q(v)|Q(v)\rangle}  \notag \\
&=& \frac{\langle Q(v/v_i)|
M_i^{\dag-1}(\partial_{\theta_i}M_i^\dag)|v_i\rangle\langle
v_i|+|v_i\rangle\langle v_i|(\partial_{\theta_i}M_i)M_i^{-1} |Q(v/v_i)\rangle%
}{\langle Q(v)|Q(v)\rangle}  \notag \\
&=& \frac{\langle Q(v/v_i)|
M_i^{\dag-1}(\partial_{\theta_i}M_i^\dag)|v_i\rangle\langle
v_i|+|v_i\rangle\langle v_i|(\partial_{\theta_i}M_i)M_i^{-1} |Q(v/v_i)\rangle%
}{\langle Q(v/v_i)|Q(v/v_i)\rangle}/\frac{\langle Q(v)|Q(v)\rangle}{\langle
Q(v/v_i)|Q(v/v_i)\rangle}  \notag \\
&=& \frac{\langle Q(v/v_i)|
M_i^{\dag-1}(\partial_{\theta_i}M_i^\dag)|v_i\rangle\langle
v_i|+|v_i\rangle\langle v_i|(\partial_{\theta_i}M_i)M_i^{-1} |Q(v/v_i)\rangle%
}{\langle Q(v/v_i)|Q(v/v_i)\rangle}/\frac{\langle Q(v/v_i)|v_i\rangle\langle
v_i|Q(v/v_i)\rangle}{\langle Q(v/v_i)|Q(v/v_i)\rangle}.
\end{eqnarray}
So computing gradient of KL-divergence of QGM reduces to preparing tensor
network state $|Q(z)\rangle$ ($z$ being empty, $v$ or $v/v_i$
respectively) and measuring the expectation value of $O=|v\rangle\langle v|$%
, $O_1=(\partial_{\theta_i}M_i)M_i^{-1}+c.c. $ and $O_2=|v_i\rangle\langle
v_i| (\partial_{\theta_i}M_i)M_i^{-1}+c.c. $. If the denominator is small,
the error of the estimation through sampling could be large, in particular when the denominator is exponentially close to $0$. But we do not need to be pessimistic. Since this denominator represents the probability of getting $v_i$ for
measuring the $i$-th variable, conditioned on the remaining variables being $%
v/v_i$, this quantity should not be small if the model distribution is close to the real data
distribution. Otherwise, we are not able to
observe this data. If the model distribution is far from the real
data distribution, there is no need to estimate the gradient precisely. Moreover,
statistical fluctuation in sampling in this case could even help to jump out of
the local minimum, which is analogous to the stochastic gradient descent method
in traditional machine learning \cite%
{shalev2014understandings,goodfellow2016deeps}.

The efficiency of estimating the expectation value of local observable
through sampling depends on the maximum value of the observable. In the case of
calculation of gradient of the KL-divergence, the maximum value depends on $M_i$ whose minimum singular value is
set to $1$. Thus, the number of sampling is
bounded by $\kappa^2/\epsilon^2$, where $\kappa$ is the condition number
which is the maximum singular value of $M_i$, $\epsilon$ is the error of the
estimation. In the case when $\kappa$ is large, $M_i$ is
ill-conditioned, resulting in bad estimation. One strategy to avoid this
situation is to use regularization term \cite%
{shalev2014understandings,goodfellow2016deeps}. For example, we may use
\begin{equation}
\sum_i\mbox{tr} M_i^\dag M_i-\chi|\det M_i|^2,
\end{equation}
as a penalty where $\chi>0$ is a hyperparameter which adjust the importance
of the second term. The first term guarantees the singular value
of $M_i$ should not be large and the second term guarantees that $%
\lambda_1\lambda_2$ should not be too small. In the case that $%
\max(\lambda_1,\lambda_2)$ is not too large, it guarantees that $%
\min(\lambda_1,\lambda_2)$ is not too small, which means $M_i$ is not too
ill-conditioned.

\section{Proof of Theorem 3}

First, we prove that $|Q(z)\rangle$ could
represent any tensor network efficiently:
\begin{lemma}\label{lem:TN}
\emph{Choosing the conditioned variables $z$ and invertible matrices properly, $\ket{Q(z)}$ could efficiently represent any tensor network in which each tensor has constant degree and the virtual index range is bounded by a constant.}
\end{lemma}

\begin{proof}
For a tensor $A_{i\cdots j}$, consider a quantum circuit preparing the following state:
\begin{equation}
\sum_{i\cdots j}A_{i\cdots j}\ket{i\cdots j}
\end{equation}
where the dimension of the Hilbert space is bounded by a constant if the number of indices and the range of each index of this tensor are bounded. The quantum circuit could be represented by a state $\ket{Q(z)}$ like the one in Fig. \ref{fig:qs} with constant size, conditioned on specific values of some variables. This is just a post-selection of measurement result in measurement-based quantum computing\cite{raussendorf2001ones}. Then we consider contracting virtual indices between different tensors. In this case, direct contraction will lead to a problem: the edge connecting two different tensors may be an identity tensor instead of a Hadamard tensor as in state $\ket Q$. We can solve this problem by introducing an extra variable in the middle of these two indices and connecting the two indices by Hadamard tensor. Further applying a Hadamard gate on it and conditioning on this extra variable being $0$, the net effect is equivalent to connecting the two indices by an identity tensor.\end{proof}



Since computation on the QGM is
reduced to preparing $|Q(z)\rangle$, we will focus on tensor network
construction for the state $|Q(z)\rangle$ in the following discussion. For an instance that
our algorithm could run in polynomial time, $\min_t\Delta_t$ and $%
\min_t\eta_t$ should be both at least $1/\mbox{poly}(n)$. We construct a tensor network satisfying this
requirement which at the same time encodes universal quantum computing. Therefore, a classical model is not able to 
to produce this result if quantum computation can not be efficiently simulated classically.

\begin{figure}[tbp]
\centering
\includegraphics[width=1\linewidth]{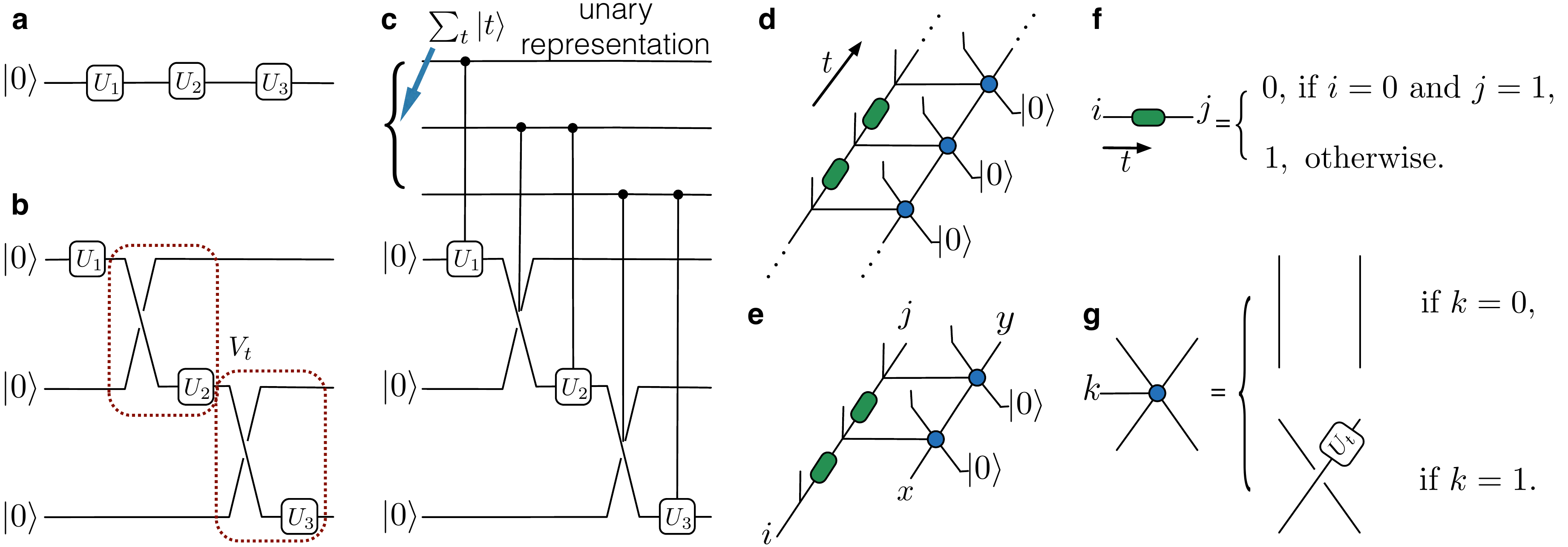}
\caption{ \textbf{Construction of the History State.} Illustration features one qubit only for convenience of drawing. It is straightforward to generalize it to the case of $n$ qubits by considering the layout of the circuits in Fig. 1 of Ref.\cite{oliveira2008complexitys}. \textbf{a,} The universal
quantum circuit. \textbf{b,} By adding swap gates and ancilla qubits, we can
simulate the circuit in \textbf{a} by a 1D (2D for $n$ qubits case) tensor
network state since there is only two gates applying on each qubit. The
circuit in red dashed box is $V_t$. \textbf{c,} The history state of circuit
\textbf{b} (here we omit the starting $|1\rangle$ and the end $|0\rangle$).
\textbf{d, } Tensor network representation of the state in \textbf{c}.
\textbf{e, }Group of tensors which is used to construct parent Hamiltonian, $%
x,y,i,j$ being virtual indices. \textbf{f,g, }Definition of tensors in \textbf{%
d,e}. \textbf{f, }Tensors representing $\sum_t|t\rangle$ (unary representation
of $t$). \textbf{g, }Tensors representing controlled-swap and controlled-$U_t$ or
simply controlled-$V_t$. }
\label{fig:clock}
\end{figure}

First, we construct the tensor network state encoding universal quantum
computation. Consider the following history state encoding history of quantum
computation:
\begin{equation}
|\psi_\text{clock}\rangle=\frac{1}{\sqrt{T+1}}\sum_{t=0}^{T}|t\rangle\otimes
V_t\cdots V_1|0\rangle^{\otimes m}
\end{equation}
where $T$ is the number of gates in the quantum circuit, $
|t\rangle=|1\cdots10\cdots0\rangle$ is the unary representation of the number $t$
(the first $t$ bits are set to $1$ and the last $T-t$ bits are set to $0$),
 $|0\rangle^{\otimes m}$ is the input state in the 2D layout of the circuits shown in Fig. 1 of Ref.\cite{oliveira2008complexitys} with $m$ (which is the number of qubits $n$ times the depth of the original circuit, i.e. roughly the same as $T$) qubits and $V_t$ is the $%
t$-th gate defined in Fig. \ref{fig:clock}b, i.e., $V_1=U_1, V_t=U_t\cdot\text{%
SWAP}$ for $t > 1$. This history state has been used to prove QMA-hardness for local Hamiltonian on a 2D
lattice \cite{oliveira2008complexitys}. This state can be represented by a ``2D" tensor network shown in Fig. 
\ref{fig:clock}d since there is only a constant number of gates applying on each
qubit. It can be verified directly that the tensor network in Fig. \ref%
{fig:clock}d represents $|\psi_\text{clock}\rangle$. For simplicity of
illustration, we consider computation on a single qubit where the ``2D" network
(actually 1D in this case) is shown in Fig. \ref{fig:clock}a. The tensor
in Fig. \ref{fig:clock}f must have the form $|1\cdots10\cdots0\rangle$
and all of them have the same weight; the left line in Fig. \ref{fig:clock}d
will be $|t\rangle$ entangled with state in right line; the tensor in
Fig. \ref{fig:clock}g is control-$V_t$ gate serving to make the right line in Fig. \ref%
{fig:clock}d  $V_t\cdots V_1|0\rangle^{\otimes m}$ if the
left line is $|t\rangle$. In this way, the state is exactly $%
|\psi_\text{clock}\rangle$. The general case ($n$-qubit case) for
universal quantum computing is constructed in the same way.

Second, we calculate the overlap between two successive tensor networks.
From Fig. \ref{fig:clock}d, direct calculation shows that the tensor network $%
|Q_t\rangle$ is
\begin{eqnarray}
\frac{1}{\sqrt{t+1}}\sum_{t_1=0}^{t}|t_1\rangle\otimes V_{t_1}\cdots
V_1|0\rangle^{\otimes m}.
\end{eqnarray}
Note that the number of bits in clock register $|t_1\rangle$ is the same for
any $t$. Then the overlap between $|Q_t\rangle$ and $|Q_{t-1}\rangle$ is
\begin{equation}
\langle Q_t|Q_{t-1}\rangle=\frac{1}{\sqrt{(t-1)t}}(t-1)=\sqrt{1-\frac{1}{t}},
\end{equation}
so
\begin{equation}
\eta_t=1-\frac{1}{t}\ge\frac{1}{2}\text{ for }t\ge2,
\end{equation}
which means
\begin{equation}
\frac{1}{\min_t\eta_t}=\mathcal{O}(1).
\end{equation}

Third, we calculate the parent Hamiltonian of tensor network $|Q_T\rangle$.
For simplicity, we consider quantum circuit on single qubit
(Fig. \ref{fig:clock}a) and the general case follows similarly. Each local term is constructed from five variables shown in Fig. \ref%
{fig:clock}e. With different virtual indices denoted by $i,j,x,y$, the
tensor corresponds to a five
dimensional space spanned by the following states:
\begin{eqnarray}
&& |000\rangle\otimes |x\rangle|0\rangle,  \notag \\
&& |111\rangle\otimes |0\rangle|0\rangle,  \notag \\
&&|100\rangle\otimes|x\rangle|0\rangle+|110\rangle\otimes|0 \rangle
U_t|x\rangle=|100\rangle\otimes|x\rangle|0\rangle+|110\rangle\otimes
V_t(|x\rangle|0\rangle),
\end{eqnarray}
where $x$ can be either $0$ or $1$. The corresponding local term in parent
Hamiltonian is projector $\Pi_t$ such that its null space is exactly the above
subspace so its rank is $32-5=27$:
\begin{eqnarray}
\Pi_t&=&\Pi_t^{(1)}+\Pi_t^{(2)}+\Pi_t^{(3)}+\Pi_t^{(4)},  \notag
\end{eqnarray}
where
\begin{eqnarray}
\Pi_t^{(1)}&=&|01\rangle\langle 01|\otimes I\otimes I\otimes
I+I\otimes|01\rangle\langle 01|\otimes I\otimes I,\quad\text{rank: }%
2^3+2^3=16,  \notag \\
\Pi_t^{(2)}&=&|000\rangle\langle 000|\otimes I \otimes |1\rangle\langle
1|+|111\rangle\langle 111|\otimes(I\otimes I-|00\rangle\langle 00|),\quad%
\text{ rank: }2+3=5,  \notag \\
\Pi_t^{(3)}&=&|100\rangle\langle 100|\otimes I\otimes |1\rangle\langle
1|+|110\rangle\langle 110|\otimes|1\rangle\langle 1|\otimes I,\quad\text{
rank: }2+2=4,  \notag \\
\Pi_t^{(4)}&=&\frac{|100\rangle\langle 100|\otimes I\otimes |0\rangle\langle
0|+|110\rangle\langle 110|\otimes |0\rangle\langle 0|\otimes
I-|100\rangle\langle 110|\otimes (|0\rangle\langle 0|\otimes
I)V_t^\dag-|110\rangle\langle 100|\otimes V_t(|0\rangle\langle 0|\otimes I)}{%
2},  \notag \\
&&\quad\text{ rank: }2,
\end{eqnarray}
where $\Pi^{(i)}_t\Pi^{(j)}_t=0$ for $i\ne j$ and the first three qubits are
in clock register. The rank of these projector could be calculated
from trace since all of them are projectors. The parent Hamiltonian is
\begin{equation}
H_p=\sum_{t=0}^T \Pi_t=H_p^{(1)}+H_p^{(2)}+H_p^{(3)}+H_p^{(4)}=\sum_{t=0}^T
\Pi_t^{(1)}+\sum_{t=0}^T \Pi_t^{(2)}+\sum_{t=0}^T \Pi_t^{(3)}+\sum_{t=0}^T
\Pi_t^{(4)},
\end{equation}
The terms at the boundary are a little bit different since $x=0,i=1$ at the
start and $j=0$ at the end.

Finally, we analyze the energy gap of the parent Hamiltonian of $|Q_t\rangle$%
. In order to analyze spectrum of $H_p$, it is convenient to introduce
the perturbation theory used in \cite{kempe2006complexitys}. We will
use first order perturbation. Consider $\widetilde H=H+V$, let $\Pi_-$ be
the projector to subspace of $H$ with $0$ eigenvalue, and the eigenvalues in $%
\Pi_+$ are greater than $J$. If $J\gg\parallel V\parallel$ where $%
\parallel\cdot\parallel$ denotes the spectral norm, then the low energy
spectrum (those much smaller than $J$) of $\widetilde H$ will be
approximately the same as the spectrum of $V_{--}\equiv \Pi_{-}V\Pi_{-}$.
For convenience, given an operator $X$, we denote $X_{-+}=\Pi_-X\Pi_+$.
Similarly for $X_{--}$ and $X_{++}$. If $X$ is block diagonal in the
subspace $\Pi_-$ and $\Pi_+$, denote $X_-=X_{--}$.

\begin{lemma}[Ref.\protect\cite{kempe2006complexitys}]
\emph{\ Consider the resolvent of $\widetilde H$
\begin{equation*}
\widetilde G(z)=(zI-\widetilde H)^{-1}
\end{equation*}
and define self-energy as
\begin{equation*}
\Sigma_-(z)=zI_--\widetilde G^{-1}_{--}(z).
\end{equation*}
Let $\widetilde \lambda_j$ be the $j$-th eigenvalue of $\widetilde H$ below $%
J/4$. The $j$-th eigenvalue of $\Sigma_-(\widetilde \lambda_j)$ is also $%
\widetilde \lambda_j$.}
\end{lemma}

\noindent Then we can prove the low energy spectrum of $\widetilde H$ and $%
V_{--}$ are approximately the same using the series expansion of $\Sigma_-(z)
$, which is the generalization of Projection Lemma in \cite%
{kempe2006complexitys}.

\begin{lemma}[Generalized Projection Lemma]
\emph{\ If the energy gap of $H$ is $J>\parallel V\parallel^2+2\parallel
V\parallel$, the spectrum of $\widetilde H$ below $J/4$ is $\parallel
V\parallel^2/J$-close to the spectrum of $V_{--}$.}
\begin{proof}
Consider the series expansion of $\Sigma_-(z)$ as shown in Ref.\cite{kempe2006complexitys}:
\begin{equation}
\Sigma_-(z)=H_- + V_{--} + V_{-+}G_+V_{+-} + V_{-+}G_+V_{++}G_+V_{+-} + V_{-+}G_+V_{++}G_+V_{++}G_+V_{+-} + \cdots
\end{equation}
where $G$ is the resolvent of $H$:
\begin{equation*}
G(z)=(zI-H)^{-1},
\end{equation*}
and $H_-=0$ in our problem. For $z\le J/4$, we have
\begin{equation}
\parallel G_+\parallel \le  \frac{4}{3}\cdot\frac{1}{J}=\mathcal O\left(\frac{1}{J}\right).
\end{equation}
Thus, for $z<J/4$ and $J> \parallel V\parallel^2+2\parallel V\parallel$, we have
\begin{equation}
\parallel \Sigma_-(z)-V_{--}\parallel=\mathcal O\left(\sum_{k=1}^\infty \frac{\parallel V\parallel ^{k+1}}{J^k} \right)=\mathcal O\left(\frac{\parallel V \parallel^2}{J}\right).
\end{equation}
According to a special case of Weyl's inequality or the one in \cite{kempe2006complexitys}, the difference between $j$-th eigenvalues of $\Sigma_-(z)$ and $V_{--}$ for any $z<J/4$ is bounded by the spectral norm of their difference. Therefore the difference between the $j$-th eigenvalue below $J/4$ of $\widetilde H$ and $V_{--}$ is $\mathcal O(\parallel V\parallel^2/J)$.
\end{proof}
\end{lemma}

\noindent Using this lemma, we prove the following lemma regarding gap of $H_p$.

\begin{lemma}
\emph{\ The ground state of $H_p$ is unique and has eigenvalue zero. If $T=%
\mbox{poly}(n)$, the energy gap of $H_p$ is at least $1/\mbox{poly}(n)$.}
\begin{proof}
Consider the spectrum of $JH_p$.
\begin{equation}
JH_p\ge J_1H_p^{(1)}+J_3H_p^{(3)}+H_p^{(4)}
\end{equation}
where we choose $J_1$ and $J_3$ later satisfying $J>J_1>J_3$. Notice that all the Hamiltonians is positive semidefinite. We regard $H=J_1H_p^{(1)}$ and $V=J_3H_p^{(3)}+H_p^{(4)}$. In this case, $\parallel V\parallel=\mathcal O(J_3T)$. Because the ground subspace of $H_p^{(1)}$ is restricted to history state,
\begin{eqnarray}
\Pi_{t--}^{(3)}&=&\ket{t-1}\bra{t-1}\otimes I\otimes \ket1\bra1+\ket{t}\bra{t}\otimes\ket1\bra1\otimes I,\\
\Pi_{t--}^{(4)}&=&\frac{\ket{t-1}\bra{t-1}\otimes I\otimes \ket{0}\bra{0}+\ket{t}\bra{t}\otimes \ket{0}\bra{0}\otimes I-\ket{t-1}\bra{t}\otimes (\ket 0\bra 0\otimes I)V_t^\dag-\ket{t}\bra{t-1}\otimes V_t(\ket0\bra0\otimes I)}{2},\\
V_{--}&=&\sum_{t}J_3\Pi_{t--}^{(3)}+\Pi_{t--}^{(4)}.
\end{eqnarray}
The difference between low energy spectrums of $\widetilde H=H+V$ and $V_{--}$ is $\mathcal O(J_3^2T^2/J_1)$. Then consider a new Hamiltonian $\widetilde H^\prime=H^\prime+V^\prime$ where
\begin{equation}
H^\prime=J_3\sum_t \Pi_{t--}^{(3)},
\end{equation}
\begin{equation}
V^\prime=\sum_t
\frac{\ket{t-1}\bra{t-1}\otimes (I\otimes I)+\ket{t}\bra{t}\otimes (I\otimes I)-\ket{t-1}\bra{t}\otimes V_t^\dag-\ket{t}\bra{t-1}\otimes V_t}{2}.
\end{equation}
The Hamiltonian $\widetilde H^\prime$ is similar to the one used to prove QMA-hardness of local Hamiltonian problems \cite{kitaev2002classicals,aharonov2002quantums} and energy gap in universality of adiabatic quantum computing\cite{aharonov2008adiabatics}. Let $\Pi^\prime_-$ be the projector to eigenspace of $H^\prime$ with zero eigenvalue. Direct calculation shows that
\begin{equation}
V^\prime_{--}\equiv \Pi^\prime_-V^\prime\Pi^\prime_-=\sum_t \Pi_{t--}^{(4)}
\end{equation}
The difference between low energy spectrum of $\widetilde H^\prime=H^\prime+V^\prime$ and $V_{--}^\prime$ is $\mathcal O(T^2/J_3)$. Now we have
\begin{equation}
JH_p\ge H+V\underset{\mathcal O(J_3^2T^2/J_1)}{\approx}V_{--}\ge V_{--}^\prime\underset{\mathcal O(T^2/J_3)}{\approx} H^\prime+V^\prime
\end{equation}
where $\underset{\delta}{\approx}$ means the difference of the low energy spectrum is bounded by $\delta$. All the derivation ignores the boundary terms that are not difficult but just tedious to take into account.
Thus we can bound the energy gap of $H_p$ to be
\begin{equation}
\Delta_p\ge \Omega\left(\frac{1}{JT^2}\right)-\mathcal O\left(\frac{J_3^2T^2}{JJ_1}\right)-\mathcal O\left(\frac{T^2}{JJ_3}\right)
\end{equation}
Choosing $J\sim J_1,J_1\sim J_3^3$ and $J_3>>T^4$ will give
\begin{equation}
\Delta_p\ge \Omega\left(\frac{1}{T^{14}}\right)=\frac{1}{\poly(n)}.
\end{equation}

\end{proof}
\end{lemma}

\noindent The same analysis for spectral gap could be applied to any
intermediate parent Hamiltonian $H_t$. Thus
\begin{equation}
\frac{1}{\min_t\Delta_t}=\frac{1}{\mbox{poly}(n)}
\end{equation}%
.

In summary, we arrive at proof of theorem 3:
\begin{proof}
To prove the theorem for the problem of computing conditional probability is very straightforward since the computation could be reduced to measuring one qubit on $\ket{Q(z)}$ which encodes universal quantum computing in our construction.

For the problem of computing gradient of KL-divergence, we just choose $M_i$ in QGM to be parameterized by
\begin{equation}
M_i= \left(
  \begin{array}{ccc}
    \theta_i & \theta_i^\prime \\
    \theta_i^{\prime\prime} &  \theta_i^{\prime\prime\prime}\\
  \end{array}
\right)
\end{equation}
and when all the $M_i$ are unitary matrices we could have the state $\ket Q$ in Fig. \ref{fig:qs}. Then we consider the derivative of $\theta_{i^\prime}$ for a unconditioned variable and $M_{i^\prime}=I$. In this case,
\begin{equation}
O_1=(\partial_{\theta_{i^\prime}}M_{i^\prime})M_{i^\prime}^{-1}+  M_i^{\dag-1}(\partial_{\theta_{i^\prime}}M_{i^\prime}^\dag)=
 \ket0\bra0ã
\end{equation}
For the term $\log\innerp{Q}{Q}$,
\begin{equation}
\frac{\partial_{\theta_i}\innerp{Q}{Q}}{\innerp{Q}{Q}}=\bra{G}O_1\ket{G}=\frac{1}{2}.
\end{equation}
Then suppose we have only one training data $\ket{v}$ such that $\ket{Q(v)}$ is the tensor network shown in Fig. \ref{fig:clock} (where the state $\ket Q$ could be the one in Fig. \ref{fig:qs} since such a $\ket Q$ with postselection is universal and the proof of lemma \ref{lem:TN} requires only universility). The gradient of KL-divergence becomes
\begin{equation}
\frac{\partial_{\theta_i}\innerp{Q(v)}{Q(v)}}{\innerp{Q(v)}{Q(v)}}=\frac{\bra{Q(v)   }O_1 \ket{Q(v)}}{\innerp{Q(v)}{Q(v)}}.
\end{equation}
which is also measuring one qubit on $\ket{Q(z)}$ thus encoding universal quantum computing.
Suppose the acceptance probability of a BQP problem is $p_0$ (we define measuring 0 as accept),
the total gradient of KL-divergence for the parameter $\theta_{i^\prime}$ is
\begin{equation}
p_0-\frac{1}{2}.
\end{equation}
So we could estimating $p_0$ to $p_0\pm1/\poly(n)$ through estimating the gradient of KL-divergence.

Meanwhile, the algorithm for preparing $\ket{Q(z)}$ runs in polynomial time, thus we prove this theorem.
\end{proof}


\end{document}